\documentclass[12pt]{amsart}
\usepackage[margin=1in]{geometry}
\usepackage{amsfonts}
\usepackage{amssymb}
\usepackage[dvips]{graphics}
\usepackage{epsfig}
\usepackage{euscript}
\usepackage{color}
\usepackage{bbm}
\usepackage[colorlinks ,pdfstartview=FitB]{hyperref}
\hypersetup{allcolors=blue}

\newtheorem{thm}{Theorem}[section]
\newtheorem{cor}[thm]{Corollary}
\newtheorem{lem}[thm]{Lemma}

\theoremstyle{definition}

\theoremstyle{remark}
\newtheorem{remark}[thm]{Remark}
\numberwithin{equation}{section}
\numberwithin{thm}{section}

\DeclareMathOperator{\RE}{Re}
\DeclareMathOperator{\IM}{Im}

\newcommand{\R}{{\mathord{\mathbb R}}}

\newcommand{\W}{\mathcal{W}}

\newcommand{\C}{{\mathord{\mathbb C}}}

\def\idty{{\mathchoice {\mathrm{1\mskip-4mu l}} {\mathrm{1\mskip-4mu l}} %
{\mathrm{1\mskip-4.5mu l}} {\mathrm{1\mskip-5mu l}}}}

\DeclareMathOperator{\Tr}{Tr}

\DeclareMathOperator{\diag}{diag}

\begin{document}

\title[Dynamical entanglement]{Dynamical Evolution of Entanglement in  Disordered Oscillator Systems}

\author[H. Abdul-Rahman]{Houssam Abdul-Rahman}
\address{Mathematics, New York University Abu Dhabi}
\email{ha2271@nyu.edu}
\date{\today}

\begin{abstract}
We study the non-equilibrium dynamics of a
disordered quantum system consisting of harmonic oscillators in a $d$-dimensional lattice. If the system is sufficiently localized, 
we show that, starting from a broad class of initial product states that are associated with a tiling (decomposition) of the $d$-dimensional lattice, the dynamical evolution of entanglement follows an area law in all times. Moreover, the entanglement bound reveals a dependency on how the subsystems are located within the lattice  in dimensions $d\geq 2$. In particular, the entanglement grows with the maximum degree of the dual graph associated with the lattice tiling.  
\end{abstract}

\maketitle

%
%

\allowdisplaybreaks

\section{Introduction}

Understanding the time evolution of entanglement after a quantum quench is essential in determining scenarios where the transmission of quantum information is possible. For this, dynamical entanglement has received strong attention in the  condensed matter physics and quantum information theory literature, see e.g., \cite{Calabrese06, Rangamani17, Alba-Calabrese17,Alba-Calabrese18, QQ-Mitra18, QQ-Lewis-Swan19, QQ-Klobas21}. 

The importance of studying the dynamics of entanglement in the presence of disorder falls within the quest of  a better understanding of {\it many body localization} (MBL), a phenomenon that has recently attracted extensive research activity in the physics literature, see e.g.,  \cite{AbanicPepic17,Agarwaletal,Altmam-etal-15,NH} for recent reviews with extensive lists of references. The general mathematical  understanding of the MBL phenomenon is far from satisfactory. To gain further insights about MBL, it is useful to study simple (toy) models. A number of localization results have been made in recent years for some  disordered many body systems: The XY chain  \cite{HSS12, PasturSlavin15, ARS15, SimsWarzel16,ARNSS16}, and its continuum counterpart model,  the Tonks-Girardeau gas \cite{SeiringerWarzel16},   the quantum Ising model  \cite{Imbrie16, ImbrieEtal17}, the attractive XXZ chain (known as the XXZ chain in the Ising phase) \cite{EKS1, EKS2, BW1, BW2, ARFS19},  the many body Holstein model \cite{MS-Holstein-17}, and the harmonic oscillator systems. 


Transport  in the systems of harmonic oscillators coupled by springs has been studied recently in the gapped case, see  \cite{ALN2010, Cr-Eis2006, Cr-Se-Eis2008, Mat-Ish1970,Sch-Cir-Wol2006}, in the form of Lieb-Robinson estimates as well as  exponential clustering of ground-state correlations.
A more recent set of papers considered the gapless disordered case in the context of MBL, see \cite{NSS12, NSS13, ARSS17,AR18,ARSS18,BSW19}. The analysis relies on the fact that the system of coupled harmonic oscillators reduces to a (disordered) free boson system, and hence it is a completely solvable model. 

Many MBL indicators are shown for the disordered harmonic oscillator systems under consideration in this work.  
The results in \cite{NSS12, NSS13, ARSS17,BSW19}  show that for a class of such disordered systems, localization for the effective one particle Hamiltonian implies MBL for the many body system. In particular, zero velocity Lieb-Robinson bounds are established, exponential clustering for the energy eigenstates and thermal states correlations are proven; and area laws for entanglement of some gaussian states (ground and thermal states) are shown. The exponential decay of correlations of the eigenstates is established in \cite{ARSS17} and it indicates an area law for the eigenstates (at least for the low lying eigenstates), see \cite{BH13,BH15}.  \cite{ARSS18} investigates more the implication between the single body localization and the MBL by studying oscillator systems that are associated with a partially localized  effective one particle Hamiltonian and it shows a MBL phase of excited states with arbitrarily large energy density. Weak entanglement for the eigenstates, demonstrated in area laws (or logarithmic scaled area laws \cite{MullerEtal20}) for the eigenstates, is considered as an essential component in understanding the MBL phenomena. The major obstacle in studying the entanglement of the eigenstates of the harmonic oscillator systems is that they are non-gaussian states. Progress in analyzing such states was made in \cite{AR18} where a class of a non-gaussian states of (disordered) harmonic oscillators is defined, and a method that provides an exact formula for their entanglements is developed.

Recent studies seem to indicate that dynamical entanglement may not, generically, satisfy an area law. For example, \cite{PlenioEtal} considers a set of  30 uncoupled harmonic oscillators cooled down to (near) the ground state and subsequently switch on the coupling suddenly (to produce a one-dimensional chain). In stark contrast with the entanglement property of the stationary ground state of the coupled system,  entanglement seems to generate over  large distances in this dynamical setting; showing a scenario where the transmission of quantum information may be possible. A natural question to ask then: Does this behavior persist in the presence of disorder?

In this work, we analytically investigate a similar dynamical setting but for the disordered gapless case and in any dimension. In particular, we consider a finite number of (disordered) local harmonic oscillator systems, each consisting of a finite number of coupled oscillators, to be initially prepared at local thermal/ground states (with possibly different temperatures for the thermal states). Then we suddenly switch on the interaction, to drive the initial product state out of equilibrium, and we study the dynamical evolution of entanglement. This scenario is also described in the literature as the study of entanglement after a \emph{quantum quench}, see e.g., \cite{Rangamani17,  QQ-Mitra18}. The case when each local system is consisting of exactly one harmonic oscillator is a special case of our general setting. We show that the dynamical entanglement follows an area law in all times. This means  that the dynamical evolution of entanglement after coupling all local systems has an upper bound that is independent of the volume of the system, and it scales like the surface area of the initial region. Moreover, the explicit upper bound for the dynamical entanglement is independent of the initial local thermal/ground states. An observation about this estimate in $d\geq 2$ is that it depends on the geometry of the interactions among the initially non-coupled subsystems. More precisely, the entanglement grows with the maximum degree of the dual graph associated with the initially uncoupled subsystems and their subsequent couplings.

The time evolution of our initial states is a gaussian (and generally mixed) state, for which the entanglement is characterized  by correlation matrices. We build on the results of the well know approach of Vidal and Werner \cite{VidalWerner}, see also \cite{NSS13, AR18, BSW19}, to find an exact formula for the logarithmic negativity of the time evolution. As opposed to all existing analysis for similar problems, the correlation matrices are not block-diagonal, moreover, the blocks include many operators that are not uniformly bounded in the volume of the system. This makes the process of obtaining a practical upper bound of entanglement  far from being a trivial manipulation of existing techniques.

This paper is organized as follows. In Section \ref{subsec:Model} we introduce the model and the localization assumption on the effective one-particle Hamiltonian.  Then, in Section \ref{subsec:Free}, we present briefly the reduction of the harmonic oscillator systems to a free boson system. Section \ref{subsec:Result} includes a precise description of the dynamical setting, the main theorem, Theorem \ref{thm:main-result}, and a detailed discussion of the result with some extreme and special cases.

The proof of the area law spans the rest of the paper. In Section \ref{subsec:LN:Weyl} we introduce  the Weyl operators, and we explain how they relate to the correlation matrices of the time-evolved states. This forms the first step in the derivation of an exact formula, Theorem \ref{thm:LN}, for the logarithmic negativity of the dynamically evolved states, see Section \ref{subsec:LN:Formula}. We then find an upper bound for the entanglement in terms of the Schatten quasi-norms, Lemma \ref{lem:LN:upperbound}, in Section \ref{subsec:LN:UpperBound}.
Section \ref{sec:AreaLaw} finishes the proof of the main result, and it includes the crucial lemma, Lemma \ref{lem:bound}, that draws the link between the entanglement upper bound and the geometry of the initial decomposition of the system.

\section*{Acknowledgment}
The author would like to thank Robert Sims and G\"{u}nter Stolz for insightful  discussions and comments.

\section{Model and Main Results}\label{sec:Model}
\subsection{Model}\label{subsec:Model}
We study the quantum harmonic oscillator systems over the finite rectangular box $\Lambda=\left([a_1,b_1]\times[a_2,b_2]\times \ldots\times[a_d,b_d]\right)\cap\mathbb{Z}^d$, for any $d\geq 1$
\begin{equation}\label{H_Lambda}
H_\Lambda=\sum_{x\in\Lambda}(p^2_x+k_x q_x^2)+\sum_{x,y\in\Lambda;\ |x-y|=1}(q_x-q_y)^2,
\end{equation}
over the Hilbert space
\begin{equation}
\mathcal{H}_\Lambda=\bigotimes_{x\in\Lambda}\mathcal{L}^2(\mathbb{R},dq_x).
\end{equation}
Here $q_x$  is the position operator at site $x\in\Lambda$, i.e., the multiplication operator by $q_x$,  and $p_x=i\partial/\partial q_x$ is the momentum operator on $\mathcal{L}^2(\mathbb{R},dq_x)$. These unbounded operators are self-adjoint on suitable domains and they satisfy the commutation relations
\begin{equation}\label{eq:pq-com}
[q_x,q_y]=[p_x,p_y]=0\quad \text{ and } [q_x,p_y]=i\delta_{x,y}\idty,
\end{equation}
where $\delta_{x,y}=\{1 \text{ if }x=y, \text{ and } 0\text{ otherwise}\}$ is the Kronecker delta function, see e.g., \cite{ReedSimon}. Disorder in the system is added through the spring constants $\{k_x\}_{x\in\Lambda}$, where they are regarded as independent identically distributed  random variables with absolute continuous distribution $\mu$ given by bounded density $\nu$, supported in $[0,k_{\text{max}}] $ with $0<k_{\text{max}}<\infty$, i.e.,
\begin{equation}
d\mu(k_x)=\nu(k_x)dk_x, \text{ with }\|\nu\|_\infty<\infty
 \text{ and } \text{supp }\nu=[0,k_{\max}].
\end{equation}

 In (\ref{H_Lambda}), $|\cdot|$ denotes the $1$-norm distance on $\Lambda$, i.e., the interaction is over all undirected edges ${x,y}$ from $\Lambda$ that correspond to the nearest neighbor sites $x$ and $y$.

It is easy to see that $H_\Lambda$ can be written in the compact form
\begin{equation}
H_\Lambda=-\bigtriangleup+ q^T h_\Lambda q, \quad \bigtriangleup \text{ is the Laplacian}, \ q=(q_x)_{x\in\Lambda},
\end{equation}
where $h_\Lambda$ is the $d$-dimensional finite volume Anderson model defined as
\begin{equation}\label{eq:elements-h-Lambda}
\langle \delta_x, h_\Lambda\delta_y\rangle=
\begin{cases}
2 d +k_x &  \text{if } x=y\\
-1 & \text{if } |x- y|=1\\
0 & \text{elsewhere}
\end{cases}.
\end{equation}
Here $\{\delta_x, x\in\Lambda\}$ is the canonical basis of $\ell^2(\Lambda)$. Our assumption of absolute continuous distribution of the $\{k_x\}_{x\in\Lambda}$ implies that the spectrum of $h_\Lambda$ is almost surely simple, see \cite[Lemma B.1]{ARSS18}. Moreover, it is clear from (\ref{eq:elements-h-Lambda}) that
\begin{equation}\label{eq:h-bound}
\text{spec}(h_\Lambda)
\subseteq\left[\min_{x\in\Lambda}k_x, 4 d + k_{\text{max}}\right],
\end{equation}
meaning that $h_\Lambda$ is almost surely strictly positive with the deterministic bound
\begin{equation}\label{eq:h-bound}
\|h_\Lambda\|\leq \mathcal{C}_h^2:=4 d + k_{\text{max}}, 
\end{equation}
and $h^{-1}_\Lambda$ (hence $h^{-1/2}_\Lambda$) exists almost surely, but $\|h^{-1/2}_\Lambda\|$ is not uniformly bounded in the volume of the system and the disorder. A consequence of this is that the many body Hamiltonian $H_\Lambda$ is not deterministically gapped, i.e., it does not have a ground state gap, see (\ref{eq:H-spec}) below.

It is well known that the analysis of $H_\Lambda$ reduces to the study of the effective one particle Hamiltonian $h_\Lambda$, for which
 we will assume that its \emph{singular eigenfunction correlators} decays exponentially. More precisely, we assume that there exist constants $C<\infty$, $\eta>0$ and $0< s\leq 1$, independent of $\Lambda$, such that
\begin{equation}\label{def:EFC}
\mathbb{E}\left(\sup_{|u|\leq 1}|\langle\delta_x, h^{-1/2}_\Lambda u(h_\Lambda)\delta_y\rangle|^s\right)< C e^{-\eta |x-y|},
\end{equation}
for all $x,y\in\Lambda$, where $\mathbb{E}(\cdot)$ is the disorder average (with respect to the product measure $d\mathbb{P}=\prod_{x\in\Lambda}d\mu$ on $\R^\Lambda$) and $\{\delta_x\}_{x\in\Lambda}$ the canonical basis of $\ell^2(\Lambda)$. The supremum is taken over all functions $u:\R\rightarrow\C$ which satisfy the pointwise bound $|u(x)|\leq 1$.  $u(h_\Lambda)$ is defined by the functional calculus of symmetric matrices.

The eigenfunction correlators in (\ref{def:EFC}) is first introduced in \cite{NSS12}, and it holds in the following cases:
\begin{itemize}
\item[(a)] For $d=1$ and any $\nu$ with $s=1/2$, \cite[Prop A.1(c) and A.4(a)]{NSS13}.
\item[(b)] For $d\geq 1$ and large disorder with $s=1$,  \cite[Prop A.1(b) and A.3(b)]{NSS13}.
\end{itemize}
We will need to deal with less singular eigenfunction correlators than (\ref{def:EFC}), in the meaning that the singular term  $h_\Lambda^{-1/2}$ is replaced by $\idty_\Lambda$ or $h_\Lambda^{1/2}$. The decay rate of such eigenfunction correlators are known in various regimes in works related to the Anderson model, see e.g., \cite{Aizenman, Aizenman-Warzel-15, Stolz}. Starting with the decay bound in (\ref{def:EFC}), one can get similar bounds for these cases,
\begin{eqnarray}\label{def:EFC-identity}
\mathbb{E}\left(\sup_{|u|\leq 1}|\langle\delta_x,  u(h_\Lambda)\delta_y\rangle|^s\right)&=&\mathbb{E}\left(\|h^{1/2}_\Lambda\|^s\sup_{|u|\leq 1}|\langle\delta_x, h^{-1/2}_\Lambda \frac{h^{1/2}_\Lambda}{\|h^{1/2}_\Lambda\|} u(h_\Lambda)\delta_y\rangle|^s\right)\nonumber\\
&\leq& \mathcal{C}_h^s\ \mathbb{E}\left(\sup_{|u|\leq 1}|\langle\delta_x, h^{-1/2}_\Lambda u(h_\Lambda)\delta_y\rangle|^s\right)\nonumber\\
&\leq& \mathcal{C}_h^s C e^{-\eta|x-y|}.
\end{eqnarray}
And similarly, 
\begin{equation}\label{def:EFC-half}
\mathbb{E}\left(\sup_{|u|\leq 1}|\langle\delta_x,  h^{1/2}_\Lambda u(h_\Lambda)\delta_y\rangle|^s\right)\leq \mathcal{C}_h^{2s} C e^{-\eta|x-y|},
\end{equation}
for all $x,y\in\Lambda$.
\begin{remark}
All of our results could be extended to more general disordered oscillator systems than \eqref{H_Lambda}, In particular, our results apply to oscillator systems in general graphs with bounded maximum degrees, and with random masses weighing the kinetic energies $p_x^2$ or random couplings at the interactions $(q_x-q_y)^2$, see e.g.,  \cite{NSS13, BSW19}, as long as localization of the effective one particle Hamiltonian in the form \eqref{def:EFC} can be verified. We limit our discussion to the case of lattices and random spring constants $k_x$ is mostly due to the fact that this can be referenced for the Anderson model (with disordered potential).
\end{remark}
\subsection{$H_\Lambda$ as a free boson system}\label{subsec:Free} $H_\Lambda$ can be written as a disordered free boson system. In the following we present the main steps of this transformation. For more details we refer the reader to e.g.,  \cite{NSS12,NSS13, AR18}.

The real positive (almost surely) operator $h_\Lambda$ can be orthogonally diagonalized 
$
h_\Lambda=\sum_{j}\gamma_j^2|\phi_j\rangle\langle\phi_j|
$
in terms of the eigenvectors $\{\phi_j,\ j=1,\ldots, |\Lambda|\}$  of $h_\Lambda$, and the corresponding positive eigenvalues: $0<\gamma_1^2\leq \gamma_2^2\leq \ldots\leq \gamma_{|\Lambda|}^2$ in non-decreasing order. This diagonalization defines almost surely the operators $\{b_j,\ j=1,\ldots,|\Lambda|\}$ (and their adjoints $\{b_j^*\}_j$)
\begin{equation}\label{eq:pq-to-b}
b_j=\frac{1}{\sqrt{2}}\left(\gamma_j^{1/2} \phi_j^T q+i \gamma_j^{-1/2} \phi_j^T p\right),\ j=1,\ldots,|\Lambda|,\ q=(q_j)_j,\ p=(p_j)_j.
\end{equation}
A direct calculation using the commutation relations (\ref{eq:pq-com}) shows that the $b$ operators and their adjoints satisfy the {\it canonical commutation relations} (CCR)
\begin{equation}\label{CCR}
[b_j,b_k]=[b_j^*,b_k^*]=0 \text{ and } [b_j,b_k^*]=\delta_{j,k}\idty,\quad j=1,\ldots,|\Lambda|
\end{equation}
and leads to the second quantization representation of the harmonic oscillator systems $H_\Lambda$,
\begin{equation}\label{eq:H-FreeBoson}
H_\Lambda=\sum_{j=1}^{|\Lambda|}\gamma_j(2 b_j^* b_j +\idty).
\end{equation}
The CCR (\ref{CCR}) implies the existence of a unique (up to a phase) vacuum $\psi_{0,\Lambda}$ of the $b$'s operators, i.e., $b_j\psi_{0,\Lambda}=0$ for all $j=1,\ldots,|\Lambda|$. This leads to the complete diagonalization of $H_\Lambda$ with the following set of eigenvalues
\begin{equation}\label{eq:H-spec}
\text{spec}(H_\Lambda)=\left\{\sum_{k=1}^{|\Lambda|}\gamma_k(2\alpha_k+1);\ \alpha_k\in\mathbb{N}_0  \text{ for }k=1,\ldots,|\Lambda|\right\},
\end{equation}
showing that the ground state gap is $2\min_k\gamma_k=2\gamma_1$.

\subsection{A quantum quench scenario}\label{subsec:Result} A quite non-trivial class of non-equilibrium processes of quantum many body systems is quantum quenches. We start with  a quantum state (say the ground state) which is prepared with respect to one Hamiltonian $H_0$ at $t=0$ and then we suddenly change the Hamiltonian from $H_0$ to a new one $H$. The original state starts to experience the time evolution for $t>0$. Such process is called a \emph{quantum quench}, see, e.g., \cite{Rangamani17}.

Next we explain how quantum quench is considered for (disordered) oscillator systems: We start with a fixed number of non-interacting oscillator systems, each system is of the form (\ref{H_Lambda}) and it is initially in a  thermal state or in its ground state. This means that our initial state is the product state of all the corresponding  local thermal/ground states. We switch on the interactions via $H_\Lambda$ and we study the entanglement created in the new coupled system, i.e., we study the entanglement of the initial product state after a quantum quench.

More precisely, we consider the $d$-dimensional lattice 
\[
\Lambda=\Lambda^{(L)}:=[-L,L]^d\cap\mathbb{Z}^d
\]
 and we decompose it into $M$ disjoint sub-rectangular regions (we tile $\Lambda$ using $M$ rectangular regions.)
\begin{equation}\label{def:partition1}
\Lambda=\bigcup_{m=1}^M \Lambda_m.
\end{equation}
Here and in the rest of the paper, we slightly abuse notation and use $\Lambda=\Lambda^{(L)}$ to denote the $d$-dimensional lattice.
 
We consider the time evolution of initial suitable states $\varrho$ that corresponds to the decomposition (\ref{def:partition1}). The initial states are product states of any finite number of thermal/ground states of restrictions of the harmonic oscillator Hamiltonian $H_\Lambda$ to subsystems.

In particular, for each $m=1,\ldots,M$ we consider the restriction $H_{\Lambda_m}$ of the harmonic system $H_\Lambda$ to $\Lambda_m$, defined similar to (\ref{H_Lambda}).

For each $m$, let $\varrho_{m,\beta_m}$ be the thermal state of the local model $H_{\Lambda_m}$ with inverse temperature $\beta_m\in (0,\infty]$, where $\beta_m=\infty$ corresponds to the ground state density,
\begin{equation}
\varrho_{m,\beta_m}=\begin{cases}
(\Tr e^{-\beta_m H_{\Lambda_m}})^{-1}\ e^{-\beta_m H_{\Lambda_m}} & \text{ if }0<\beta_m<\infty\\
|\psi_{0,\Lambda_m}\rangle\langle\psi_{0,\Lambda_m}| &\text{ if }\beta_m=\infty
\end{cases}.
\end{equation}
We choose the initial state to be the product state for $\beta=(\beta_1,\ldots,\beta_M)\in(0,\infty]^M$ 
\begin{equation}\label{def:rho0}
\varrho_\beta=\bigotimes_{m=1}^M \varrho_{m,\beta_m},
\end{equation}
where
$(0,\infty]^M:=\{(\beta_1,\ldots,\beta_M);\ \beta_1,\ldots,\beta_M\in(0,\infty]\}$.
Here we remark that $\varrho_{\beta}$ in (\ref{def:rho0}) is stationary, i.e., time invariant, under the Hamiltonian of the non-interacting systems,
\begin{equation}
H_{0,\Lambda}=\sum_{m=1}^M H_{\Lambda_m}\otimes \idty_{\Lambda\setminus\Lambda_m}.
\end{equation}
We suddenly switch on the (spring) interactions between the local Hamiltonians $\{H_{\Lambda_m}\}_m$ to obtain the full system $H_\Lambda$ at $t>0$. 

We study the entanglement of the (Sch\"odinger) time evolution of the initially non-entangled state $\varrho_\beta$ defined in (\ref{def:rho0}).
\begin{equation}\label{eq:time-evolution}
\varrho_{t,\beta}=e^{-i t H_\Lambda}\varrho_\beta e^{i t H_\Lambda}.
\end{equation}
Towards this, we fix a subregion $\Lambda_0\subset \Lambda$
and we consider the bipartition of the Hilbert space
\begin{equation}\label{bipartition}
\mathcal{H}_\Lambda=\mathcal{H}_{1}\otimes\mathcal{H}_{2} \text{\quad with\quad }\mathcal{H}_{1}=\bigotimes_{x\in\Lambda_0}\mathcal{H}_x, \ \mathcal{H}_2=\bigotimes_{x\in\Lambda\setminus\Lambda_0}\mathcal{H}_x. 
\end{equation}

In the case when all the local subsystems are cooled down to their ground states, the initial state (\ref{def:rho0}) and its time evolution (\ref{eq:time-evolution}) are pure states for which the (von Neumann) entanglement entropy is the best tool to quantify entanglement. In any other case ($\beta_m<\infty$ for some $m$), our initial states (\ref{def:rho0}) are mixed states, for which the {\it logarithmic negativity} is considered as a suitable entanglement measure \cite{Plenio2005}: It is well known that the logarithmic negativity is an upper bound to the distillable entanglement \cite{PlenioEtal}, and for pure states, it is an upper bound to the entanglement entropy e.g., \cite{VidalWerner, NSS13}. Therefore,  we use the logarithmic negativity to bound the entanglement dynamics. By $\varrho^{T_1}$ we denote the partial transpose with respect the first component in (\ref{bipartition}) of a state $\varrho$ in $\mathcal{H}_{\Lambda}$, see the Appendix of \cite{NSS13} for a detailed discussion of partial transposes. The logarithmic negativity $\mathcal{N}(\varrho)$ of $\varrho$ is the logarithm of the trace norm of $\varrho^{T_1}$, 
\begin{equation}\label{def:LN}
\mathcal{N}(\varrho)=\log\|\varrho^{T_1}\|_1.
\end{equation}
Here $\|\cdot\|_1$ is the trace norm, i.e., $\|A\|_1=\Tr|A|$. In (\ref{def:LN}), 
we choose $\log$ to denote the natural logarithm, as opposed to the logarithm base 2 ($\log_2$) used in this context in the information theory literature. The distinction is irrelevant for our work, as we will not keep track of universal constants.

In our analysis we do not require the decomposition $\Lambda_1, \dots,\Lambda_M$ to be compatible with the bipartition of the system into $\Lambda_0$ and $\Lambda\setminus\Lambda_0$. However, if $\Lambda_0$ is chosen to be a union of  sub-regions $\Lambda_m$, then the initial state is a product state, i.e., not entangled with respect to the bipartition $\mathcal{H}_1\otimes \mathcal{H}_2$,
\begin{equation}
\mathcal{N}(\varrho_{t=0,\beta})=\mathcal{N}(\varrho_\beta)=0.
\end{equation}
For positive time $t>0$, $\varrho_{t,\beta}$  defined in (\ref{eq:time-evolution}) is entangled (not seperable), and hence, its logarithmic negativity is strictly positive.

\subsection{Results}
Our main result provides an area law for the Schr\"{o}dinger time evolution of $\varrho_{\beta}$ under to the full system $H_{\Lambda}$, with respect to the bipartition (\ref{bipartition}).  This means that the entanglement the time-evolved  state $\rho_{t,\beta}$ scales at most like the surface area $|\partial\Lambda_0|$ of $\Lambda_0$ at all times and for any choice of the initial local inverse temperatures. Here $\partial\Lambda_0$ 
\begin{equation}
\partial\Lambda_0=\{x\in\Lambda_0;\ \exists y\in\Lambda\setminus\Lambda_0 \text{ with } |x-y|=1\}
\end{equation}
denotes the boundary of $\Lambda_0$. 

While our entanglement bound does not depend explicitly on the number of subregions $M$ in (\ref{def:partition1}), it depends on the geometry of the decomposition  in the meaning of how the subregions are ``distributed'' within the lattice $\Lambda$. In particular, we show an area law with a pre-factor that is proportional to the maximum degree of the \emph{dual graph} of the tiling (\ref{def:partition1}). More precisely, 
we understand the decomposition (\ref{def:partition1}) as a graph 
 \begin{equation}\label{def:Graph}
 \mathcal{G}_M=(\mathcal{V}_M,\mathcal{E}_M),
 \end{equation}
  with the set of vertices $\mathcal{V}_M$ consisting of the subregions $\{\Lambda_m\}$ and $\mathcal{E}_M$ is the set of undirected edges between neighboring subregions (vertices in $\mathcal{V}_M$), i.e.,
\begin{equation}\label{def:Graph-VE}
\mathcal{V}_M=\{\Lambda_m;\ m=1,\ldots,M\} \text{  and } 
 \mathcal{E}_M=\left\{\{\Lambda_j,\Lambda_k\};\ d_\Lambda(\Lambda_j,\Lambda_k)=1\right\},
\end{equation}
where $d_\Lambda(\Lambda_j,\Lambda_k)$ is the lattice distance between the two subsets, $\Lambda_j$ and $\Lambda_k$, of $\Lambda$, i.e.,
\begin{equation}
 d_\Lambda( \Lambda_j,\Lambda_k)=\min\{|x-y|;\ x\in\Lambda_j,y\in\Lambda_k\}.
 \end{equation}
 The graph $\mathcal{G}_M=(\mathcal{V}_M,\mathcal{E}_M),$ defined in (\ref{def:Graph}) and (\ref{def:Graph-VE}) is referred to as the dual graph of the decomposition of $\Lambda$ into $M$ disjoint subregion (\ref{def:partition1}).

Our entanglement bound depends on the \emph{maximum degree} $\Delta(\mathcal{G}_M)$ of the dual graph $\mathcal{G}_M$,
\begin{equation}\label{def:R-2}
 \Delta(\mathcal{G}_M)=\max_{j=1,\ldots,M}\left|\left\{\Lambda_k;\ \{\Lambda_j,\Lambda_k\}\in\mathcal{E}_M \right\}\right|\leq M-1.
\end{equation}
 
With $\mathbb{E}(\cdot)$ denotes the disorder average, we will prove the following theorem about the logarithmic negativity for $\varrho_{t,\beta}$.
\begin{thm}\label{thm:main-result}
Assume that the effective one particle Hamiltonian $h_\Lambda$ satisfies (\ref{def:EFC}). 
Consider the initial state $\varrho_\beta$ defined in (\ref{def:rho0}) with respect to the decomposition $(\ref{def:partition1})$, and its Schr\"{o}dinger evolution, $\varrho_{t,\beta}=e^{-itH_{\Lambda}}\varrho_\beta e^{itH_\Lambda}$ under the full harmonic oscillator systems $H_{\Lambda}$.

Then there exists $C'<\infty$ independent of $\Lambda$ and $M$ such that
\begin{equation}\label{eq:result}
\mathbb{E}\left(\sup_{t\in\R,\ \beta\in(0,\infty]^M}\mathcal{N}(\varrho_{t,\beta})\right)\leq C' (1+\Delta(\mathcal{G}_M)^{s/4})\ |\partial \Lambda_0|
\end{equation}
for all $1\leq m\leq |\Lambda|$ and the choices of the sub-boxes $\{\Lambda_m\}_m$. In (\ref{eq:result}), $\Delta(\mathcal{G}_M)$ is defined in (\ref{def:R-2}), and $s\in(0,1]$ is the constant in the eigenfunction correlators (\ref{def:EFC}). The supremum is taken over all times $t\in \mathbb{R}$ and all thermal/ground states $\varrho_{m,\beta_m}$ of $H_{\Lambda_m}$, $m=1,\ldots,M$.
\end{thm}
The constant $C'$ in Theorem \ref{thm:main-result}  depends on the dimension $d$ of the lattice $\Lambda$, $k_{\max}$, and on the constants in the eigenfunction correlators (\ref{def:EFC}): $s\in(0,1]$, $C$, and $\eta$. In particular $C'$ can be chosen as
\begin{equation}\label{def:C'}
C'=\frac{288}{s} \mathcal{C}_h^s C\left(\sum_{x\in\mathbb{Z}^d}e^{-\frac{1}{4} \eta|x|}\right)^3,
\end{equation}
where $\mathcal{C}_h$ is the $h_\Lambda^{1/2}$ norm-bound (\ref{eq:h-bound}). 


The area law (\ref{eq:result}), giving an upper bound for the dynamical evolution of entanglement proportional to the surface area of the subsystem $\Lambda_0$, is uniform not only in time $t\in\R$, the size of the system $\Lambda$, and the subsystem $\Lambda_0$, but it also applies uniformly to all possible products of thermal/ground states of $H_{\Lambda_m}$ for $m=1,\ldots,M$ irrespective of their corresponding temperatures. This result adds to the many other MBL indicators shown for the disordered quantum harmonic oscillator models that satisfy the localization assumption (\ref{def:EFC}), \cite{NSS12,NSS13,ARSS17,AR18,ARSS18,BSW19}.

In the following we comment on the maximum degree $\Delta(\mathcal{G}_M)$ in the entanglement bound (\ref{eq:result})
\begin{itemize}
\item $\Delta(\mathcal{G}_M)$ is a genuine high dimensional object: While there is essentially one dual graph (a chain) associated with the  decomposition of a one-dimensional chain into $M$ sub-intervals, and hence $\Delta(\mathcal{G}_M)\leq 2$, the situation in higher dimensions is  totally different. For example, one can tile a 2-dimensional rectangular region using $M$ sub-rectangles in many  ways that correspond to different dual graph's maximum degrees $\Delta(\mathcal{G}_M)$. In particular, horizontal (or vertical) slicing of $\Lambda$ corresponds to $\Delta(\mathcal{G}_M)=2$, and any other tiling (using rectangles) corresponds to higher graph's maximum degree, see the special  case discussed in comment (iii) below, where $\Delta(\mathcal{G}_M)$ is the surface area of $\Lambda_0$.

\item One crucial case to discuss is the behavior of the dynamical entanglement bound in the thermodynamic limit. This leads to the interesting case when $\Delta(\mathcal{G}_M)$ does not grow as $M\rightarrow\infty$, this is happening for example when the dual graph is translation invariant. In this case, one may discuss the thermodynamic limit with a growing (without bound) number of initial fixed size subsystems. Of course, the thermodynamic limit $\Lambda\rightarrow\mathbb{Z}^d$ (i.e., $L\rightarrow\infty$) can also be understood with a fixed number of decomposition $M$, but we have to be careful on how to control the growing subregions.

\item Whether the entanglement bound (\ref{eq:result}) is sharp or not, is an open interesting question. Our bound is only an upper bound for the dynamical entanglement. In particular,  it  is pressing to investigate more the dependency of the entanglement upper bound on $\Delta(\mathcal{G}_M)$, for example, by proving an entanglement lower bound that grows with $\Delta(\mathcal{G}_M)$. In this case, our bound would reveal an important insight about the dynamical evolution of entanglement in high dimensions: Dynamical entanglement depends on the way the subregions are tiling (or covering) the lattice. That this theoretical insight may be used to maximize (or tune up) entanglement in two or more dimensional quantum models is a question that needs more theoretical  and numerical investigations. 

\end{itemize}

We comment on some special cases in the following remarks:
\begin{itemize}\itemsep0.2cm
\item[(i)] In the extreme case where $M=1$ (meaning that $\Delta(\mathcal{G}_1)=0$), our main result reduces to an area law for the  entanglement of (equilibrium) thermal states and the ground state of $H_\Lambda$ in (\ref{H_Lambda}), making the area law results in \cite{NSS13,BSW19} a special case  of Theorem \ref{thm:main-result}.

\item[(ii)] In the other extreme case where each subsystem consists of only one site, i.e., $M=|\Lambda|$, the initial Hamiltonian $H_{0,\Lambda}$ is a system of non-coupled  oscillators over the $d$-dimensional lattice $\Lambda$,
\begin{equation}
H_{0,\Lambda}=\sum_{x\in\Lambda} H_{\{x\}}\otimes\idty_{\Lambda\setminus \{x\}},
\end{equation}
where $H_{\{x\}}$ is the Hamiltonian of a single quantum harmonic oscillator
\begin{equation}\label{def:Hx}
H_{\{x\}}=p_x^2+k_x q_x^2.
\end{equation}
In this case, the scenario is as follows: We consider $|\Lambda|$ quantum oscillators  $H_{\{x\}}$ placed on the lattice $\Lambda$ and prepared at a corresponding thermal state $\varrho_{\beta_x}$, or cooled down to the ground state ($\beta_x=\infty$). Then we couple all next neighbor local oscillators by springs suddenly, and study the entanglement with respect to the bipartition (\ref{bipartition}).  Note that here the dual graph associated with this decomposition is the whole lattice with next neighbors interactions, i.e., $\mathcal{G}_{|\Lambda|}=(\mathcal{V}_{|\Lambda|},\mathcal{E}_{|\Lambda|})$ with vertices $\mathcal{V}_{|\Lambda|}=\Lambda$ and edges $\mathcal{E}_{|\Lambda|}=\big\{\{x,y\};\ |x-y|=1\big\}$, and observe that in this case $\Delta(\mathcal{G}_{|\Lambda|})=2d$. Hence, Theorem \ref{thm:main-result} gives directly the following corollary.
\begin{cor}\label{cor:1}
Assume that $h_\Lambda$ satisfies (\ref{def:EFC}). Consider the thermal (or ground) state $\varrho_{\beta_x}$  of the quantum  harmonic oscillator $H_{\{x\}}$ in (\ref{def:Hx}) for each $x\in\Lambda$. Then
\begin{equation}
\mathbb{E}\left(\sup_{t\in\R,\ \beta=(\beta_x)\in(0,\infty]^{|\Lambda|}}\mathcal{N}\big(e^{-itH_\Lambda}\bigotimes_{x\in\Lambda}\varrho_{\beta_x}e^{itH_\Lambda}\big)\right)\leq C' (1+(2d)^{s/4})\ |\partial \Lambda_0|
\end{equation}
where $C'$ is given in $(\ref{def:C'})$.
\end{cor}

 It is noteworthy that numerics in \cite{PlenioEtal} suggest that for the counterpart one dimensional gapped model, dynamical entanglement in the same situation addressed in Corollary \ref{cor:1} (only local ground states are considered in \cite{PlenioEtal}.) is generated over very large distances, indicating that the entanglement might not be following an area law. If this is the case, then Corollary \ref{cor:1} shows that disordered un-gapped harmonic oscillator systems generate weaker (dynamical) entanglement than their gapped counterparts.

\item[(iii)] It is also remarkable, that one can carefully construct decompositions (\ref{def:partition1}) for which the maximum degree of the associated graph $\Delta(\mathcal{G}_M)$, that appears in the entanglement upper bound (\ref{eq:result}), is equal to $|\partial\Lambda_0|$. This can be done for example if $\Lambda_0$ is chosen such that $\Lambda_0=\Lambda_1$ is a large enough   sub-rectangular region (so that $|\partial\Lambda_0|\gg2d$) away from the boundaries of the whole lattice $\Lambda$ (within $\mathbb{Z}^d$), and all other $\Lambda_j$'s are just of cardinality one, i.e., $M=|\Lambda|-|\Lambda_0|+1$. In this case, 
\begin{equation}\label{ex:G}
\Delta(\mathcal{G}_M)=\max_{j=1,\ldots,M}\left|\{\Lambda_k;\ d_\Lambda (\Lambda_j,\Lambda_k)=1\}\right|=
\left|\{\Lambda_k;\ d_\Lambda (\Lambda_0,\Lambda_k)=1\}\right|=|\partial\Lambda_0|,
\end{equation}
and hence, we've crafted a special decomposition for which the entanglement upper bound is $\mathcal{O}(|\partial\Lambda_0|^{1+\frac{s}{4}})$, this follows from (\ref{eq:result}), which also means that
\begin{equation}\label{}
\mathbb{E}\left(\sup_{t\in\R,\ \beta\in(0,\infty]^M}(\mathcal{N}(\varrho_{t,\beta}))^{\frac{4}{s+4}}\right)\leq \tilde{C} \ |\partial \Lambda_0|,
\end{equation}
where we used Jensen's inequality $\mathbb{E}(\cdot^q)\leq \mathbb{E}(\cdot)^q$ for any $q\in(0,1]$. Without ruling out the possibility that for this special case the bound is not optimal, we observe that
this is an area law if we accept that $(\mathcal{N}(\varrho_{t,\beta}))^{\frac{4}{s+4}}$, where $0<s\leq1$ is an entanglement measure as well. We think that any $q$-power, for $q\in(0,1]$,  of the logarithmic negativity $(\mathcal{N} (\varrho))^q$ is an entanglement quantifier for the following reasons: it
vanishes when $\varrho$ is separable, and more importantly, it is \emph{entanglement monotone} under general \emph{positive partial transpose} preserving (known as PPT) operations (it does not increase on average  under general PPT operations). The latter follows directly from the proof of the case $q=1$ presented in \cite{Plenio2005} by replacing the logarithm function by $q$-power of the logarithm in the argument after equations (7) in \cite{Plenio2005}, noting that $(\log(x))^q$ is concave and monotone increasing for $0<q\leq 1$.
\end{itemize}


\section{Logarithmic Negativity}\label{sec:LN}
In this section, we derive an exact formula for the logarithmic negativity of $\rho_{t,\beta}$, see Theorem \ref{thm:main-result}, then we find a practical entanglement upper bound, Lemma \ref{eq:LN:upperbound-0}. 

\subsection{Weyl operators expectations and correlation matrices}\label{subsec:LN:Weyl}

It is well known that the thermal and ground states of the free boson systems are gaussian states (quasi-free), for which exact entanglement formulas are derived with the aid of the Weyl operator expectations, often called the \emph{(quantum) characteristic functions}, see e.g., \cite{VidalWerner, NSS13, AudenaertEtal2002, EisertEtal2010}.

In the following, we define the Weyl operator and find its expectation at the time-evolved state $\varrho_{t,\beta}$ defined in (\ref{eq:time-evolution}). Along the way, we show that the initial state $\varrho_\beta$ and its time evolution $\varrho_{t,\beta}$ are both gaussian states. In general, the product and the time evolution of gaussian states are gaussian. 

Towards writing the explicit formula for the logarithmic negativity of $\varrho_{t,\beta}$, we start by defining the Weyl operators (also known as displacement operators). For every $z\in\mathbb{C}$  the corresponding Weyl operator $\W_z$ is defined as the unitary operator
\begin{equation}
\W_z=\exp\left(i(\RE[z] q_x+ \IM[z] p_x)\right)
\end{equation}
Then for every $f=(f_x)_{x\in\Lambda}\in\ell^2(\Lambda)$, the Weyl operator is defined as 
\begin{equation}\label{def:Weyl}
\W(f)=\bigotimes_{x\in\Lambda}\W_{f_x}=\exp\left(i\sum_{x\in\Lambda}(\RE[f_x]q_x+\IM[f_x]p_x)\right).
\end{equation}
Let us identify $\ell^2(\Lambda;\C)$ with $\ell^2(\Lambda;\R)\oplus\ell^2(\Lambda,\R)$, i.e.,
\begin{equation}
f\in\ell^2(\Lambda;\C)\sim \tilde{f}=\left(\begin{array}{c}\RE f\\ \IM f\end{array}\right)\in \ell^2(\Lambda;\R)\oplus\ell^2(\Lambda,\R)=:\ell^2(\Lambda,\R)^{\oplus 2}.
\end{equation}
Explicit calculations show that for any gaussian state (quasi-free) $\varrho$ on $\mathcal{H}_\Lambda$ and any $f\in\ell^2(\Lambda)$, the Weyl operator expectations is given as
\begin{equation}\label{def:gaussian}
\langle\W(f)\rangle_{\varrho}=\exp\left(-\frac{1}{4}\langle\tilde{f},\Gamma_{\varrho}\tilde{f}\rangle\right), \text{ where } \Gamma_\varrho=
\begin{pmatrix}
\langle q q^T \rangle_{\varrho} & \langle q p^T \rangle_{\varrho}\\
\langle p q^T \rangle_{\varrho} & \langle p p^T \rangle_{\varrho}
\end{pmatrix}.
\end{equation}

Here $\Gamma_\varrho$ is the $2|\Lambda|\times 2|\Lambda|$ (hermitian non-negative) position-momentum correlation matrix. For a definition of general gaussian states on the CCR algebra, we refer the reader to, e.g., \cite{Bratteli-Robinson, quasi-free2, quasi-free3, NSS13, BSW19}.

In the case when $\varrho$ is the product state $\varrho_\beta$ in (\ref{def:rho0}), we obtain
\begin{eqnarray}\label{Weyl:first}
\langle\W(f)\rangle_{\varrho_\beta} = \prod_{m=1}^M \langle\W\left(\tilde{f}|_{\Lambda_m}\right)\rangle_{\varrho_{m,\beta_m}}
&=&
\prod_{m=1}^M \exp\left(-\frac{1}{4}\left\langle \tilde{f}|_{\Lambda_m},\Gamma_{\varrho_{m,\beta_m}} \tilde{f}|_{\Lambda_m}\right\rangle\right) \nonumber\\
&=&
\exp\left(-\frac{1}{4}\left\langle\oplus_{m} \tilde{f}|_{\Lambda_m},\bigoplus_m\Gamma_{\varrho_{m,\beta_m}}\left(\oplus_m \tilde{f}|_{\Lambda_m}\right)\right\rangle\right) 
\end{eqnarray}
where we used the fact that $\varrho_{m,\beta_m}$ is a gaussian state, and hence its Weyl operator expectation is characterized by  the correlation matrix $\Gamma_{\varrho_{m,\beta_m}}$ as in (\ref{def:gaussian}).
\begin{equation}
\Gamma_{\varrho_{m,\beta_m}}=
\begin{pmatrix}\
\left\langle (q)_{\Lambda_m} (q)_{\Lambda_m}^T \right\rangle_{\varrho_{m,\beta_m}} & \left\langle (q)_{\Lambda_m} (p)_{\Lambda_m}^T \right\rangle_{\varrho_{m,\beta_m}}\\
\left\langle (p)_{\Lambda_m} (q)_{\Lambda_m}^T \right\rangle_{\varrho_{m,\beta_m}} & \left\langle (p)_{\Lambda_m} (p)_{\Lambda_m}^T \right\rangle_{\varrho_{m,\beta_m}}
\end{pmatrix},
\end{equation}
where here and in the following we use the notation $(\cdot)_{\Lambda_m}$ to denote the restriction to $\ell^2(\Lambda_m)$. 
In (\ref{Weyl:first}), $\tilde{f}|_{\Lambda_m}$ denotes the restriction to $\ell^2(\Lambda_m;\R)^{\oplus 2}$, and the direct sums in (\ref{Weyl:first}) correspond to $\bigoplus_m \ell^2(\Lambda_m;\R)^{\oplus 2}$.

A simple reorder of  basis in (\ref{Weyl:first}) maps $\oplus_m\tilde{f}|_{\Lambda_m}$ to $ \tilde{f}$ and gives
\begin{equation}\label{eq:corr:product}
\langle\W(f)\rangle_{\varrho_\beta} =
\exp\left(-\frac{1}{4}\left\langle\tilde{f},\Gamma_{\varrho_\beta}\tilde{f}\right\rangle\right) 
\end{equation}
where
\begin{equation}\label{def:corr:rho-beta}
\Gamma_{\varrho_\beta}=
\begin{pmatrix}\
\bigoplus_{m}\left\langle (q)_{\Lambda_m} (q)_{\Lambda_m}^T \right\rangle_{\varrho_{m,\beta_m}} & \bigoplus_{m}\left\langle (q)_{\Lambda_m} (p)_{\Lambda_m}^T \right\rangle_{\varrho_{m,\beta_m}}\\
\bigoplus_{m}\left\langle (p)_{\Lambda_m} (q)_{\Lambda_m}^T \right\rangle_{\varrho_{m,\beta_m}} & \bigoplus_{m}\left\langle (p)_{\Lambda_m} (p)_{\Lambda_m}^T \right\rangle_{\varrho_{m,\beta_m}}
\end{pmatrix}.
\end{equation}
A direct inspection shows that $\Gamma_{\varrho_\beta}$ in (\ref{def:corr:rho-beta}) is the correlation matrix of the product state $\varrho_\beta=\bigotimes_m\varrho_{m,\beta_m}$. So (\ref{eq:corr:product}) shows that $\varrho_\beta$ (or generally, the tensor product of gaussian states) is a gaussian state.
In our case we have, see e.g., \cite{NSS12, Bratteli-Robinson} or \cite[Chap. XII.12]{Messiah1999}, defined almost surely,
\begin{eqnarray}\label{eq:corr:values}
\left\langle (q)_{\Lambda_m} (q)_{\Lambda_m}^T \right\rangle_{\varrho_{m,\beta_m}}&=&\coth(\beta_m h_{\Lambda_m}^{1/2})h_{\Lambda_m}^{-1/2}, 
\\
\left\langle (p)_{\Lambda_m} (p)_{\Lambda_m}^T \right\rangle_{\varrho_{m,\beta_m}}&=& \coth(\beta_m h_{\Lambda_m}^{1/2}) h_{\Lambda_m}^{1/2}, \nonumber
\\
\left\langle (q)_{\Lambda_m} (p)_{\Lambda_m}^T \right\rangle_{\varrho_{m,\beta_m}}&=&\left\langle (p)_{\Lambda_m} (q)_{\Lambda_m}^T \right\rangle_{\varrho_{m,\beta_m}}=0. \nonumber
\end{eqnarray}

Since the Hamiltonian generating the dynamics $H_\Lambda$ is the particularly simple form of free bosons in terms of the $b_j$ operators, one finds
\begin{equation}\label{eq:tau-b}
\tau_t^{H_\Lambda}(b_j)=e^{-it\gamma_j} b_j \text{ and } \tau_t^{H_\Lambda}(b^*_j)=e^{it\gamma_j}b^*_j,
\end{equation}
where $\tau_t^{H_\Lambda}(b_j)$ is the Heisenberg dynamics of the (operator) $b_j$ given as $\tau_t^{H_\Lambda}(b_j)=e^{it H_\Lambda} b_j e^{-it H_\Lambda}$.

The $H_\Lambda$ dynamics (\ref{eq:tau-b}) with (\ref{eq:pq-to-b}) imply that for any state $\varrho$, the $\varrho_t$ expectation of the Weyl operator (where $\varrho_t=e^{-itH_\Lambda}\varrho e^{itH_\Lambda}$ is the Schr\"odinger time evolution of $\varrho$ with respect to $H_\Lambda$.) is given by the formula, see also e.g., \cite[Thm 5.2.8 (4)]{Bratteli-Robinson},
 \begin{equation}\label{eq:Weyl:dynamics}
 \langle\W(f)\rangle_{\varrho_t}=\left\langle\W\big(\tilde{f}_t\big)\right\rangle_{\varrho}, \text{ where } \tilde{f}_t=E_t  \tilde{f}
 \end{equation}
and $E_t :\ell^2(\Lambda,\R)^{\oplus 2}\rightarrow \ell^2(\Lambda,\R)^{\oplus 2}$ is the mapping that generates the dynamics (defined almost surely),
\begin{equation}\label{def:E_t }
E_t =
\begin{pmatrix}
\cos(2t h_\Lambda^{1/2}) & -h_\Lambda^{1/2}\sin(2t h_\Lambda^{1/2})\\
h_\Lambda^{-1/2}\sin(2t h_\Lambda^{1/2}) & \cos(2t h_\Lambda^{1/2})
\end{pmatrix}.
\end{equation}
The Weyl operator expectations in (\ref{eq:corr:product}) and (\ref{eq:Weyl:dynamics}), and the position-momentum correlations (\ref{eq:corr:values}) give instantly that the Weyl operator expectation at $\varrho_{t,\beta}$ is
\begin{equation}\label{eq:Weyl:dyanmics2}
\langle\W(f)\rangle_{\varrho_{t,\beta}}=\exp\left(-\frac{1}{4}\left\langle\tilde{f},\Gamma_{\varrho_{t,\beta}}\tilde{f}\right\rangle\right)
\end{equation}
where 
\begin{equation}\label{def:corr:rho-t-beta}
\Gamma_{\varrho_{t,\beta}}=E_t ^T \Gamma_{\varrho_{\beta}} E_t 
=E_t ^T
\begin{pmatrix}
\bigoplus_{m=1}^M \coth(\beta_m h_{\Lambda_m}^{1/2})h_{\Lambda_m}^{-1/2} & 0\\
0 & \bigoplus_{m=1}^M \coth(\beta_m h_{\Lambda_m}^{1/2})h_{\Lambda_m}^{1/2}
\end{pmatrix}
E_t .
\end{equation}
Here $\Gamma_{\varrho_{t,\beta}}$ represents the position-momentum correlation matrix of $\varrho_{t,\beta}$. So (\ref{eq:Weyl:dyanmics2}) shows that $\varrho_{t,\beta}$ is a gaussian state (or more generally, the time evolution of a gaussian state is gaussian).

We note here that  the mapping $E_t $ is symplectic, i.e., 
\begin{equation}\label{def:symp}
E_t  J E_t ^T=
J \text{ and } E_t^{-1}=J^T E_t^T J\text{ where }J=
\begin{pmatrix}
0 & -\idty_\Lambda\\
\idty_\Lambda & 0
\end{pmatrix},
\end{equation}
giving the block form of $E_t ^{-1}$
\begin{equation}\label{def:Tt-inv}
E_t ^{-1}=\begin{pmatrix}
\cos(2t h_\Lambda^{1/2}) & h_\Lambda^{1/2}\sin(2t h_\Lambda^{1/2})\\
-h_\Lambda^{-1/2}\sin(2t h_\Lambda^{1/2}) & \cos(2t h_\Lambda^{1/2})
\end{pmatrix},
\end{equation}
and observe that the inverse of the correlation matrix is (defined almost surely)
\begin{equation}\label{def:Gamma-beta-t-inv}
\Gamma_{\varrho_{t,\beta}}^{-1}= E_t^{-1} 
\begin{pmatrix}
\bigoplus_{m=1}^M \tanh(\beta_m h_{\Lambda_m}^{1/2})h_{\Lambda_m}^{1/2} & 0\\
0 &\bigoplus_{m=1}^M \tanh(\beta_m h_{\Lambda_m}^{1/2})h_{\Lambda_m}^{-1/2}
\end{pmatrix}
 (E_t^T)^{-1}.
\end{equation}

\subsection{A formula for the logarithmic negativity}\label{subsec:LN:Formula}
Here we provide an exact formula for the logarithmic negativity of the time-evolved state $\varrho_{t,\beta}$ given in (\ref{eq:time-evolution}), see Theorem \ref{thm:LN}. 

To ease notation in the following, we introduce the mapping  $\Pi_{Q}$  on $\mathcal{B}(\ell^2(\Lambda,\R)^{\oplus 2})$ for every orthogonal $Q\in\mathcal{B}(\ell^2(\Lambda,\R)^{\oplus 2})$, defined as
\begin{equation}\label{def:pi}
\Pi_Q(A)=Q^T A Q.
\end{equation}
Due to the identification 
$
\ell^2(\Lambda;\C)\sim \ell^2(\Lambda,\R)^{\oplus 2},
$
and the system decomposition $\Lambda=\Lambda_0 \cup \Lambda\setminus\Lambda_0$  with the corresponding Hilbert space bipartition  $\mathcal{H}_\Lambda=\mathcal{H}_1\otimes\mathcal{H}_2$ given in (\ref{bipartition}), some direct sums are understood to be acting on
$
\left(\ell^2(\Lambda_0)\oplus \ell^2(\Lambda\setminus\Lambda_0)\right)^{\oplus 2}.
$

\begin{thm}\label{thm:LN}
Fix $\Lambda_0\subset\Lambda$, and consider the initial product state $\varrho_\beta$ of  local thermal states (\ref{def:rho0}) and its Schr\"{o}dinger time evolution $\varrho_{t,\beta}=e^{-itH_{\Lambda}} \varrho_\beta e^{itH_{\Lambda}}$ under the full harmonic oscillator systems $H_\Lambda$.

Then the logarithmic negativity of $\varrho_{t,\beta}$ with respect to the decomposition (\ref{bipartition}) is given by the formula (almost surely)
\begin{equation}\label{eq:LN-1}
\mathcal{N}(\varrho_{t,\beta})=\frac{1}{4}\Tr\left[\chi_{(1,\infty)}(\Upsilon)\log \Upsilon\right] \text{ where } 
\Upsilon:=\Gamma_{\varrho_{t,\beta}}^{-1/2} \Pi_{\widetilde{\mathbb{P}}\circ J}(\Gamma_{\varrho_{t,\beta}}^{-1})\Gamma_{\varrho_{t,\beta}}^{-1/2}>0
\end{equation}
where 
\begin{itemize}
\item we use the standard notation $\chi_{(1,\infty)}(\Upsilon)$ to denote the spectral projection of $\Upsilon$ onto the interval $(1,\infty)$, i.e., the orthogonal projection onto the subspace $\Upsilon>1$.
\item $\Gamma_{\varrho_{t,\beta}}^{-1}$ is the inverse of the $\varrho_{t,\beta}$-correlation matrix, it is given by the equations (\ref{def:Gamma-beta-t-inv}) and (\ref{def:Tt-inv}).
\item $\widetilde{\mathbb{P}}=\mathbb{P}\oplus \mathbb{P}$ on $\mathcal{B}(\ell^2(\Lambda,\R)^{\oplus 2})$, and $\mathbb{P}=(-\idty_{\Lambda_0})\oplus \idty_{\Lambda\setminus\Lambda_0}$ with respect to $\ell^2(\Lambda_0)\oplus\ell^2(\Lambda\setminus\Lambda_0)$.
\item $\Pi_{\widetilde{\mathbb{P}}\circ J}(\cdot)$ is defined in (\ref{def:pi}) and $\widetilde{\mathbb{P}}\circ J$ is the composition of $\widetilde{\mathbb{P}}$ and $J$ given in (\ref{def:symp}).
\end{itemize}
\end{thm}

\begin{proof}
The first part of the proof goes along the lines of \cite[Proof of Thm 4.1]{AR18} or \cite[Proof of Thm 3.4]{NSS13}. We present the essential steps here for completeness.

A direct calculation shows that the expectation of the Weyl operator at $\varrho_{t, \beta}^{T_1}$ is given by
\begin{equation}
\langle\W(f)\rangle_{\varrho_{t, \beta}^{T_1}}=\left\langle\W\left(
(\idty_\Lambda\oplus  \mathbb{P})
\tilde{f}
\right)\right\rangle_{\varrho_{t, \beta}}=\exp\left(-\frac{1}{4}
\left\langle\tilde{f},
\Pi_{\idty_\Lambda\oplus\mathbb{P}}(\Gamma_{\varrho_{t,\beta}})
\tilde{f}\right\rangle\right)
\end{equation}
where $\mathbb{P}$ is the diagonal operator given in Theorem \ref{thm:LN}.

Since $\Pi_{\idty_\Lambda\oplus\mathbb{P}}(\Gamma_{\varrho_{t,\beta}})$ is real symmetric and positive definite (almost surely), then by the Williamson Theorem, see, e.g., \cite[Thm 8.11]{Williamson}, there exists a $2|\Lambda|\times 2|\Lambda|$ symplectic $S$ such that
\begin{equation}
S^T \Pi_{\idty_\Lambda\oplus\mathbb{P}}(\Gamma_{\varrho_{t,\beta}})S
=D\oplus D
, \text{ where } D=\diag\{d_1,\ldots,d_{|\Lambda|}\}.
\end{equation}
Here $d_j>0$ for all $j$, are the symplectic eigenvalues of $\Pi_{\idty_\Lambda\oplus\mathbb{P}}(\Gamma_{\varrho_{t,\beta}})$, which are the positive eigenvalues of 
\begin{equation}\label{def:Y}
Y:=i \left(\Pi_{\idty_\Lambda\oplus\mathbb{P}}(\Gamma_{\varrho_{t,\beta}})\right)^{1/2} J \left(\Pi_{\idty_\Lambda\oplus\mathbb{P}}(\Gamma_{\varrho_{t,\beta}})\right)^{1/2}.
\end{equation}

Symplectic $S$ induces a unitary $U\in\mathcal{B}(\mathcal{H}_\Lambda)$ such that, see \cite{Symp1,Symp2}.
\begin{equation}
\langle\W(f)\rangle_{U^* \varrho_{t,\beta}^{T_1} U}=\langle\W(S \tilde{f})\rangle_{\varrho_{t,\beta}}
=\exp\left(-\frac{1}{4}\left\langle\tilde{f},
(D\oplus D)\tilde{f}\right\rangle\right)=e^{-\frac{1}{4}\sum_{j=1}^{|\Lambda|}d_j|f_j|^2}.
\end{equation}
\cite[Thm 3.2]{AR18} or \cite[Lemma 3.5]{NSS13} defines explicitly the operators $\{\varrho_{d_j}\}_j$ with 
\begin{equation}\label{eq:norm1-rhoj}
\|\varrho_{d_j}\|_1=\begin{cases}
1 & \text{if } d_j\geq 1 \\
1/d_j & \text{if } d_j<1
\end{cases},
\end{equation}
for which
\begin{equation}
\langle\W(f)\rangle_{\otimes_j \varrho_{d_j}}=e^{-\frac{1}{4}\sum_{j=1}^{|\Lambda|}d_j|f_j|^2}
=
\langle\W(f)\rangle_{U^* \varrho_{t,\beta}^{T_1} U} \text{ for all }f\in\ell^2(\Lambda).
\end{equation}
This provides the unitary decomposition of $\varrho_{t,\beta}^{T_1}$, see \cite[Lemma 3.1]{NSS13}
\begin{equation}
\varrho_{t,\beta}^{T_1}= U \bigotimes_{j=1}^{|\Lambda|}\varrho_j U^*.
\end{equation}
Meaning that $\|\varrho_{t,\beta}^{T_1}\|_1=\prod_{j}\|\varrho_j\|_1$, then (\ref{eq:norm1-rhoj}) gives that the logarithmic negativity (\ref{def:LN}) is given by the formula
\begin{equation}
\mathcal{N}(\varrho_{t,\beta})=\sum_{j;\ d_j<1}\log d_j^{-1}=\Tr\left[\chi_{(1,\infty)}(Y^{-1})\log Y^{-1}\right]
\end{equation}
where it follows from (\ref{def:Y}) and $J^{-1}=J^T=-J$ that
\begin{equation}
Y^{-1}=i \Pi_{\idty_\Lambda\oplus\mathbb{P}}(\Gamma_{\varrho_{t,\beta}}^{-1/2}) J \Pi_{\idty_\Lambda\oplus\mathbb{P}}(\Gamma_{\varrho_{t,\beta}}^{-1/2}).
\end{equation}

Note that $Y^*=Y$ and
$Y^T=-Y$ meaning that $Y$, and hence $Y^{-1}$, has a symmetric spectrum about zero. Thus $\mathcal{N}(\varrho_{t,\beta})$ can be written as
\begin{equation}
\mathcal{N}(\varrho_{t,\beta})=\frac{1}{2}\Tr[\chi_{(1,\infty)}((Y^{-1})^2)\log \sqrt{(Y^{-1})^2}]=\frac{1}{4}\Tr\left[\Pi_{\idty_\Lambda\oplus\mathbb{P}}\big(\chi_{(1,\infty)}((Y^{-1})^2)\big)
\log \Pi_{\idty_\Lambda\oplus\mathbb{P}}\big((Y^{-1})^2\big)\right].
\end{equation}
The desired formula for the logarithmic negativity (\ref{eq:LN-1}) follows by observing that 
\begin{equation}
\Pi_{\idty_\Lambda\oplus\mathbb{P}}\left(\chi_{(1,\infty)}((Y^{-1})^2)\right)= \chi_{(1,\infty)}\left(\Pi_{\idty_\Lambda\oplus\mathbb{P}}((Y^{-1})^2)\right),
\end{equation}
and that
$\Pi_{\idty_\Lambda\oplus\mathbb{P}}\big((Y^{-1})^2\big)=\Upsilon$ defined in (\ref{eq:LN-1}).
\end{proof}

\subsection{An upper bound}\label{subsec:LN:UpperBound}
Recall that the main object in the logarithmic negativity formula (\ref{eq:LN-1}) is the inverse of the $\varrho_{t,\beta}$-correlation matrix $\Gamma_{\varrho_{t,\beta}}^{-1}$, given by the formula
\begin{equation}\label{eq:Gamma-beta-t-inv-2}
\Gamma_{\varrho_{t,\beta}}^{-1}=E_t^{-1} \begin{pmatrix}
     \bigoplus_{j=1}^M h_{\Lambda_m}^{\frac{1}{2}}\tanh(\beta_m h_{\Lambda_m}^{1/2}) & 0 \\
     0 & \bigoplus_{m=1}^M h_{\Lambda_m}^{-\frac{1}{2}}\tanh(\beta_m h_{\Lambda_m}^{1/2}) \\
   \end{pmatrix} \left(E_t^{-1}\right)^T
\end{equation}
where  $E_t $ is defined in (\ref{def:E_t }), and we re-state $E_t ^{-1}$ here for the   reader's convenience,
\begin{equation}\label{def:Tt-inv}
E_t ^{-1}=\begin{pmatrix}
\cos(2t h_\Lambda^{1/2}) & h_\Lambda^{1/2}\sin(2t h_\Lambda^{1/2})\\
-h_\Lambda^{-1/2}\sin(2t h_\Lambda^{1/2}) & \cos(2t h_\Lambda^{1/2})
\end{pmatrix}.
\end{equation}
The main technical difference between the entanglement formulas of equilibrium states in the literature, see e.g., \cite{VidalWerner, NSS13, AR18, BSW19} and the formula of entanglement here is that the (inverse) correlation matrix  for our class of non-equilibrium states is not block-diagonal. This is due to the time evolution that is characterized by $E_t$.  Moreover, having entries with no a-priori upper bounds, that are uniform in the disorder and in the volume of the system,  in the formula of $\Gamma_{\varrho_{t,\beta}}^{-1}$, i.e., $h_\Lambda^{-1/2}$ and $\bigoplus_m h_{\Lambda_m}^{-1/2}$, is making the process of deriving an upper bound of the logarithmic negativity, that yields eventually an area law, far from being trivial. 

The plan is to tame the unbounded terms using the deterministic facts: Lemma \ref{lem:bound} and (\ref{eq:h-bound}), and using the eigenfunction correlators (\ref{def:EFC}) after averaging the disorder. 

To this end, we direct our attention to the term $\bigoplus_m h_{\Lambda_m}^{-1/2}$ in the second diagonal block in the formula of inverse correlation matrix $\Gamma_{\varrho_{t,\beta}}^{-1}$ in (\ref{eq:Gamma-beta-t-inv-2}). 

\begin{remark}\label{rem:beta}
While it is straight forward to see  that
$
h_{\Lambda_m}^{-\frac{1}{2}}\tanh(\beta_m h_{\Lambda_m}^{1/2})\leq \beta_m \idty_{\Lambda_m}$ for every $m=1,\ldots,M$
giving that
\begin{equation}\label{eq:DNU}
\big\|\bigoplus_m h_{\Lambda_m}^{-\frac{1}{2}}\tanh(\beta_m h_{\Lambda_m}^{1/2}) \big\| \leq \max_{m}\beta_m;
\end{equation}
such bound yields an area law with a pre-factor that is increasing in the maximal inverse temperature $\max_m\beta_m$, which diverges in the case of having at least one initial local ground state in the product state. Our entanglement bound is independent of the initial (inverse) temperatures associated with the product states. In the process of obtaining such entanglement bound, we do not use the bound (\ref{eq:DNU}), and instead we benefit from the bound on the eigenfunction correlators (\ref{def:EFC}), (\ref{def:EFC-identity}), and  (\ref{def:EFC-half}) that have bounds independent of the function $|u|\leq 1$.
\end{remark}

We write $\Gamma_{\varrho_{t,\beta}}^{-1}$ as a product of two matrices, one takes the bounded term $\bigoplus_m h_{\Lambda_m}^{1/2}$ and the other takes the unbounded part $\bigoplus_m h_{\Lambda_m}^{-1/2}$. So,
\begin{equation}
\Gamma_{\varrho_{t,\beta}}^{-1}=M^{(1)}_{t,\beta}M^{(2)}_{t,\beta}, \text{ and hence } \Pi_{J}\big(\Gamma_{\varrho_{t,\beta}}^{-1}\big)=\Pi_{J}(M^{(1)}_{t,\beta})\Pi_{J}(M^{(2)}_{t,\beta})
\end{equation}
where
\begin{eqnarray}
M^{(1)}_{t,\beta}&:=& E_t ^{-1}
\begin{pmatrix}
     \bigoplus_{j=1}^M h_{\Lambda_m}^{\frac{1}{2}}\tanh(\beta_m h_{\Lambda_m}^{1/2}) & 0 \\
     0 & \idty_\Lambda \\
   \end{pmatrix}, \label{def:M1}\\
M^{(2)}_{t,\beta}&:=&\begin{pmatrix}
     \idty_\Lambda & 0 \\
     0 & \bigoplus_{m=1}^M h_{\Lambda_m}^{-\frac{1}{2}}\tanh(\beta_m h_{\Lambda_m}^{1/2}) \\
   \end{pmatrix}
(E_t ^{-1})^T \label{def:M2}.
\end{eqnarray}
We observe that these essential matrices are $2\times 2$-block matrices,  for which we use the notation $(\cdot)_{i,j}$ to denote the $ij$-th block for $i,j\in\{1,2\}$.
Then we can prove the following bound for $\mathcal{N}(\varrho_{t,\beta})$,
\begin{lem}\label{lem:LN:upperbound}
For any $\alpha\in(0,1]$, the logarithmic negativity of $\varrho_{t,\beta}$ given by the formula (\ref{eq:LN-1}) in Theorem \ref{thm:LN} has the following upper bound
\begin{equation}\label{eq:LN:upperbound-0}
\mathcal{N}(\varrho_{t,\beta})\leq \frac{1}{\alpha}\sum_{i,j=1}^2\left\|\left([M^{(2)}_{t,\beta},\widetilde{\mathbb{P}}]\Pi_{J}(M^{(1)}_{t,\beta})\right)_{i,j}\right\|_\alpha^\alpha, \text{ for any }\alpha\in(0,1],
\end{equation}
where $\|\cdot\|_{\alpha}=\left(\Tr|\cdot|^{\alpha}\right)^{1/\alpha}$. (It is the trace norm when $\alpha=1$ and it denotes the Schatten $\alpha$-quasi-norm when $\alpha\in(0,1)$.) $M^{(1)}_{t,\beta}$ and $M^{(2)}_{t,\beta}$ are given by the formulas (\ref{def:M1}) and (\ref{def:M2}); respectively.
\end{lem}
We note here that $\|A\|_\alpha=\left(\Tr|A|^{\alpha}\right)^{1/\alpha}$ is a matrix norm if and only if $\alpha\geq 1$. For $0<\alpha<1$, $\|\cdot\|_\alpha$ is denoted by the \emph{Schatten $\alpha$-quasi-norm}, it is absolutely homogeneous and positive definite, but it does not satisfy the triangle inequality.
\begin{proof}
We will prove the following bound for $\mathcal{N}(\varrho_{t,\beta})$,
\begin{equation}\label{eq:LN:upperbound}
\mathcal{N}(\varrho_{t,\beta})\leq \frac{1}{2\alpha}\left\|\left[M^{(2)}_{t,\beta},\widetilde{\mathbb{P}}\right]\Pi_{J}(M^{(1)}_{t,\beta})\right\|_{\alpha}^{\alpha},
\end{equation}
then we use the fact that for any $2\times 2$-block matrix $A=\big((A)_{i,j}\big)$, the $\alpha$-power of the Schatten $\alpha$-quasi-norm $\|A\|_\alpha^\alpha$ is bounded in terms of the sum of the $\|\cdot\|_\alpha^\alpha$ of the four blocks of $A$ as follows
\begin{equation}\label{lem:A}
\|A\|_\alpha^\alpha\leq 2\sum_{i,j=1}^2\|(A)_{i,j}\|_\alpha^\alpha,
\end{equation}
to get the desired bound (\ref{eq:LN:upperbound-0}).

(\ref{lem:A}) can be seen by first writing the block matrix $A$ as
\begin{equation}\label{eq:A-decomp}
A=
(A)_{1,1} \oplus (A)_{2,2}
+
\begin{pmatrix}
0 & (A)_{1,2}\\
(A)_{2,1} & 0
\end{pmatrix}.
\end{equation}
Then observe that 
\begin{equation}
\|(A)_{1,1} \oplus (A)_{2,2}\|_\alpha^\alpha=\|(A)_{1,1}\|_\alpha^\alpha +\| (A)_{2,2}\|_\alpha^\alpha
\end{equation}
and that
\begin{equation}
\begin{pmatrix}
0 & (A)_{1,2}\\
(A)_{2,1} & 0
\end{pmatrix}
\ \begin{pmatrix}
0 & \idty\\
\idty & 0
\end{pmatrix}=
(A)_{1,2} \oplus (A)_{2,1},
\end{equation}
hence,  
\begin{equation}
\left\|
\begin{pmatrix}
0 & (A)_{1,2}\\
(A)_{2,1} & 0
\end{pmatrix}
\right\|_\alpha^\alpha=\|(A)_{1,2}\|_\alpha^\alpha +\| (A)_{2,1}\|_\alpha^\alpha.
\end{equation}
Inequality (\ref{lem:A}) follows directly by applying the well known inequality, see e.g.  \cite[Thm 7.8]{Weidmann}
\begin{equation}\label{eq:alpha-norm-bound-0}
\|A_1+A_2\|_\alpha^\alpha\leq 2\|A_1\|_\alpha^\alpha+ 2\|A_2\|_\alpha^\alpha.
\end{equation}

In the following, we prove the bound (\ref{eq:LN:upperbound}).

First, we will write the formula of the logarithmic negativity (\ref{eq:LN-1}) in terms of the following non-negative monotone increasing function on $[0,\infty)$
\begin{equation}
\phi(x):=\left\{
           \begin{array}{ll}
             0 & \hbox{if } 0\leq x\leq 1 \\
             \log x & \hbox{if } x>1
           \end{array}
         \right..
\end{equation}
For later application of inequalities due to Weyl and Horn \cite[Thm 1.15]{Simon}, we note that $x\mapsto \phi(e^x)$ is convex. 

Formula (\ref{eq:LN-1}) can be written in terms of the function $\phi$ and the matrices $M^{(1)}_{t,\beta}$ and $M^{(2)}_{t,\beta}$ defined in (\ref{def:M1}) and (\ref{def:M2}), respectively, as
\begin{eqnarray}\label{eq:LN:2}
4\mathcal{N}(\varrho_{t,\beta})&=&\sum_{j}\phi\left(\lambda_j\left(\Gamma_{\varrho_{t,\beta}}^{-1}\widetilde{\mathbb{P}}
\Pi_J\big(\Gamma_{\varrho_{t,\beta}}^{-1}\big)
 \widetilde{\mathbb{P}}\right)\right) \nonumber\\
 &=& 
 \sum_{j}\phi\left(\lambda_j\left(M^{(2)}_{t,\beta} \widetilde{\mathbb{P}}
\Pi_J(M^{(1)}_{t,\beta})\Pi_J(M^{(2)}_{t,\beta}) \widetilde{\mathbb{P}} M^{(1)}_{t,\beta}\right)\right),
\end{eqnarray}
where $\lambda_1(\cdot)\geq \lambda_2(\cdot)\geq\ldots\geq \lambda_{|\Lambda|}(\cdot)$ denote the eigenvalues in decreasing order. In (\ref{eq:LN:2}) we used the simple fact that $AB$ and $BA$ have the same non-zero eigenvalues.

Then we use the following well known inequalities
\begin{equation}\label{def:ineq:Weyl-Horn}
\sum_j \varphi(|\lambda_j(A_1A_2)|)\leq \sum_j \varphi(\sigma_j(A_1A_2))\leq \sum_j \varphi(\sigma_j(A_1)\sigma_j(A_2))
\end{equation}
for any compact operators $A_1$ and $A_2$, and every function $\varphi$ that is non-negative monotone increasing function on $[0,\infty)$ so that  $t\mapsto\varphi(e^t)$ is convex. Here $\sigma_1(\cdot)\geq \sigma_2(\cdot)\geq\cdots\geq \sigma_{|\Lambda|}(\cdot)$ denote the singular values in decreasing order. In (\ref{def:ineq:Weyl-Horn}), the first inequality is due to Weyl \cite{Weyl}, see also \cite[Thm 1.15]{Simon}, and the second is Horn inequality \cite{Horn},

Apply inequality (\ref{def:ineq:Weyl-Horn}) in (\ref{eq:LN:2}) with $A_1=M^{(2)}_{t,\beta} \widetilde{\mathbb{P}}\Pi_J(M^{(1)}_{t,\beta})$, $A_2=\Pi_J(M^{(2)}_{t,\beta}) \widetilde{\mathbb{P}} M^{(1)}_{t,\beta}$, and $\varphi=\phi$ to obtain
\begin{eqnarray}
4\mathcal{N}(\varrho_{t,\beta})
&\leq& \sum_j\phi\left(\sigma_j\left(M^{(2)}_{t,\beta} \widetilde{\mathbb{P}}
\Pi_J(M^{(1)}_{t,\beta})\right)\sigma_j\left(\Pi_J(M^{(2)}_{t,\beta}) \widetilde{\mathbb{P}} M^{(1)}_{t,\beta}\right)\right)\nonumber \\
&=& 2 \sum_j\phi\left(\sigma_j\left(M^{(2)}_{t,\beta} \widetilde{\mathbb{P}}
\Pi_J(M^{(1)}_{t,\beta})\right)\right)
\end{eqnarray}
where in the last step we used the fact
$
\Pi_J(M^{(2)}_{t,\beta}) \widetilde{\mathbb{P}} M^{(1)}_{t,\beta}=\Pi_J\left(
M^{(2)}_{t,\beta} \widetilde{\mathbb{P}}
\Pi_J(M^{(1)}_{t,\beta})
\right)
$
meaning that $M^{(2)}_{t,\beta} \widetilde{\mathbb{P}}
\Pi_J(M^{(1)}_{t,\beta})$ and $\Pi_J(M^{(2)}_{t,\beta}) \widetilde{\mathbb{P}} M^{(1)}_{t,\beta}$ have the same singular values. We also used the fact that $\phi(x^2)=2\phi(x)$ for all $x\geq 0$.

Next we rewrite the operator $M^{(2)}_{t,\beta}\widetilde{\mathbb{P}} \Pi_J(M^{(1)}_{t,\beta})$ as 
\begin{equation}
M^{(2)}_{t,\beta}\widetilde{\mathbb{P}} \Pi_J(M^{(1)}_{t,\beta})=[M^{(2)}_{t,\beta},\widetilde{\mathbb{P}}]\Pi_J(M^{(1)}_{t,\beta})+
\widetilde{\mathbb{P}} M^{(2)}_{t,\beta} \Pi_J(M^{(1)}_{t,\beta})
\end{equation}
and we use Fan inequality \cite{Fan51} (see also \cite[Thm 1.7]{Simon}): for any $n\geq 0$ and any compact operators $A_1, A_2$.
\begin{equation}
 \sigma_{n+1}(A_1+A_2)\leq  \sigma_{n+1}(A_1)+ \sigma_1(A_2), 
 \end{equation}
to obtain,
\begin{equation}\label{eq:LN-bound-5}
2\mathcal{N}(\varrho_{t,\beta})\leq \sum_j \phi\left(\sigma_j\left([M^{(2)}_{t,\beta},\widetilde{\mathbb{P}}]\Pi_J(M^{(1)}_{t,\beta})\right)+\sigma_1\left(\widetilde{\mathbb{P}} M^{(2)}_{t,\beta} \Pi_J(M^{(1)}_{t,\beta})\right)\right).
\end{equation}
Here
\begin{equation}
M^{(2)}_{t,\beta} \Pi_J(M^{(1)}_{t,\beta})=\begin{bmatrix}
                 \idty_\Lambda & 0 \\
                 0 & \bigoplus_{m=1}^M \tanh^2(\beta_m h_{\Lambda_m}^{1/2}) \\
               \end{bmatrix}\leq \idty_\Lambda^{\oplus 2},
\end{equation}
hence, 
\begin{equation}\label{eq:s1}
\sigma_1(\widetilde{\mathbb{P}} M^{(2)}_{t,\beta} \Pi_J(M^{(1)}_{t,\beta}))=\left\| M^{(2)}_{t,\beta} \Pi_J(M^{(1)}_{t,\beta})
\right\|\leq 1.
\end{equation}
Moreover, note that $\phi$ is increasing and satisfies 
\begin{equation}\label{eq:log-bound}
\phi\big(x+1\big)=\log(x+1)\leq \frac{1}{\alpha}\log(x^\alpha+1)\leq  \frac{1}{\alpha} x^\alpha
\end{equation}
for every $x\geq 0$ and   $\alpha\in(0,1]$. Now, use (\ref{eq:s1}) and (\ref{eq:log-bound}) in (\ref{eq:LN-bound-5}),
\begin{equation}
\mathcal{N}(\varrho_{t,\beta})\leq \frac{1}{2 \alpha} \sum_j \sigma_j^\alpha\left([M^{(2)}_{t,\beta},\widetilde{\mathbb{P}}]\Pi_J(M^{(1)}_{t,\beta})\right),
\end{equation}
which shows the bound (\ref{eq:LN:upperbound-0}) and finishes the proof of the lemma.
\end{proof}

\section{An Area Law}\label{sec:AreaLaw}

The starting point for the proof of an area law will be the bound (\ref{eq:LN:upperbound-0}) in Lemma \ref{lem:LN:upperbound}. Take the supremum over time $t\in\R$ and inverse temperatures $\beta\in(0,\infty]^M$, then average the disorder to obtain the bound
\begin{equation}\label{eq:LN-bound-E}
\mathbb{E}\left(\sup_{t, \beta} \mathcal{N}(\varrho_{t,\beta})\right)\leq \frac{1}{\alpha}\sum_{i,j=1}^2\mathbb{E}\left(\sup_{t,\beta}\left\|\left([M^{(2)}_{t,\beta},\widetilde{\mathbb{P}}]\Pi_{J}(M^{(1)}_{t,\beta})\right)_{i,j}\right\|_\alpha^\alpha\right), \text{ for any }\alpha\in(0,1].
\end{equation}

To prove the main result (Theorem \ref{thm:main-result})  we need to list the explicit formulas for the four $|\Lambda|\times|\Lambda|$ blocks of $[M^{(2)}_{t,\beta},\widetilde{\mathbb{P}}]\Pi_{J}(M^{(1)}_{t,\beta})$, in which we use the 
 following shorts
\begin{equation}\label{def:Ds}
\mathcal{D}^+_\beta:=\bigoplus_{m=1}^M h_{\Lambda_m}^{\frac{1}{2}}\tanh(\beta_m h_{\Lambda_m}^{1/2}) \text{ \quad   and \quad}
\mathcal{D}^-_\beta:=\bigoplus_{m=1}^M h_{\Lambda_m}^{-\frac{1}{2}}\tanh(\beta_m h_{\Lambda_m}^{1/2})
\end{equation}
 for the  direct sums appearing in the formulas of $M^{(1)}_{t,\beta}$ and $M^{(2)}_{t,\beta}$ in (\ref{def:M1}) and (\ref{def:M2}), respectively; and we note that since $|\tanh(x)|\leq 1$ the operator $\mathcal{D}^+_\beta$  has the following (deterministic) bound
  \begin{equation}\label{eq:D+bound}
 \|\mathcal{D}^+_\beta\| \leq \max_{m} \|h_{\Lambda_m}^{1/2}\|\leq \|h_\Lambda^{1/2}\|\leq \mathcal{C}_h=\sqrt{4d+k_{\max}} \text{ see (\ref{eq:h-bound}).}
 \end{equation}
 
 While $\mathcal{D}^-_\beta$ is not uniformly bounded in the disorder, the volume of the system, or the inverse temperature;  the next lemma shows that $\mathcal{D}^-_\beta h_\Lambda^{1/2}$ is deterministically bounded by a constant that depends only on the maximal degree of the dual graph associated with the initial decomposition of the system, see also Remark \ref{rem:beta}.
 
Note that  $\mathcal{D}^-_\beta h_\Lambda^{1/2}$ encapsulates the main features of the dual graph $\mathcal{G}_M=(\mathcal{V}_M,\mathcal{E}_M)$ defined in (\ref{def:Graph}) and (\ref{def:Graph-VE}): The decomposition  of $\Lambda$ (\ref{def:partition1}), associated with $\mathcal{V}_M$, is reflected in $\mathcal{D}^-$ and the interaction between the corresponding systems, associated with the edges $\mathcal{E}_M$, is described through $h_\Lambda$.

 \begin{lem}\label{lem:bound}
For  the $d$-dimensional finite volume Anderson model $h_\Lambda$ defined in (\ref{eq:elements-h-Lambda}), and the direct sum $\mathcal{D}^-_\beta$ given in (\ref{def:Ds}), we have the bound
\begin{equation}\label{eq:D-h-bound}
\left\|\mathcal{D}^-_\beta h_\Lambda^{1/2}\right\|\leq \sqrt{\Delta(\mathcal{G}_M)+1}
\end{equation}
where $\Delta(\mathcal{G}_M)$ is the maximum degree of the graph $\mathcal{G}_M$ in (\ref{def:Graph}), and it is given by the formula (\ref{def:R-2}).
\end{lem}

\begin{proof}
We bound the norm
\begin{equation}
\left\|\mathcal{D}^-_\beta h_\Lambda \mathcal{D}^-_\beta \right\|= \left\|\mathcal{D}^-_\beta h_\Lambda^{1/2}\right\|^2.
\end{equation}
First, note that $\mathcal{D}^-_\beta$ can be written as
\begin{equation}
\mathcal{D}^-_\beta=\left(\bigoplus_{m}\tanh(\beta_m h_{\Lambda_m}^{1/2})\right)
\left(\bigoplus_m h_{\Lambda_m}^{-1/2}\right).
\end{equation}
Then since $|\tanh(x)|\leq 1$ for all $x\in\R$, then it is enough to show that
\begin{equation}
\left\|Z\right\|\leq \Delta(\mathcal{G}_M)+1, \text{ where } Z:=\left(\bigoplus_m h_{\Lambda_m}^{-1/2}\right)\ h_{\Lambda}\left(\bigoplus_m h_{\Lambda_m}^{-1/2}\right).
\end{equation}
Next we write $h_\Lambda$ in the basis $\bigoplus_m\ell^2(\Lambda_m)$ that corresponds to the graph $\mathcal{G}_M=(\mathcal{V}_M,\mathcal{E}_M)$ in (\ref{def:Graph}), we recall here that
\begin{equation}
\mathcal{E}_M=\big\{\{\Lambda_i,\Lambda_j\};\ d_\Lambda(\Lambda_i,\Lambda_j)=1\big\},
\end{equation}
 i.e., we understand $h_\Lambda$ as $M\times M$-block matrix, with $(h_\Lambda)_{i,j}=\iota_{\Lambda_i}^*h_\Lambda\iota_{\Lambda_j}\in\R^{|\Lambda_i|\times|\Lambda_j|}$ where $\iota_{\Lambda_i}$ is the canonical embedding $\iota_{\Lambda_i}:\R^{|\Lambda_i|}\rightarrow\R^{|\Lambda|}$. Note that $\iota_{\Lambda_j}^*h_\Lambda\iota_{\Lambda_j}=h_{\Lambda_j}$ for $j=1,\ldots,M$, and  that $(h_\Lambda)_{i,j}=0$ if $\{\Lambda_j,\Lambda_k\}\not\in \mathcal{E}_M$, i.e., if $d_\Lambda (\Lambda_j,\Lambda_k)\neq 1$. Meaning that,
\begin{equation}
h_\Lambda=\bigoplus_{m=1}^M h_{\Lambda_m}
+
\sum_{{\tiny \begin{array}{c}
i,j=1,\ldots,M\\
\{\Lambda_i,\Lambda_j\}\in\mathcal{E}_M\\
\end{array}}} E_{ij}\otimes (h_\Lambda)_{i,j},
\end{equation}
where $E_{ij}$ is the $M\times M$ matrix with 1 in the $ij$-th entry and zeros elsewhere.
Multiply $h_\Lambda$ from right and left by $\bigoplus_m h_{\Lambda_m}^{-1/2}$
to obtain the following formula for $Z$,
\begin{equation}\label{eq:Z:1}
Z-\idty_\Lambda=\sum_{{\tiny \begin{array}{c}
i,j=1,\ldots,M\\
\{\Lambda_i,\Lambda_j\}\in\mathcal{E}_M\\
\end{array}}} E_{ij}\otimes h_{\Lambda_i}^{-1/2} (h_\Lambda)_{i,j} h_{\Lambda_j}^{-1/2}=\sum_{{\tiny \begin{array}{c}
i,j=1,\ldots,M\\
\{\Lambda_i,\Lambda_j\}\in\mathcal{E}_M\\
\end{array}}} E_{ij}\otimes (Z)_{i,j}.
\end{equation}
We next show that
\begin{equation}\label{eq:norm-bound-2}
\Big\|Z-\idty_\Lambda \Big\|
\leq \Delta(\mathcal{G}_M)
\end{equation}
to get the desired bound (\ref{eq:D-h-bound}).

Since $h_\Lambda\geq 0$ almost surely,  then $Z\geq 0$ (almost surely) by construction, and hence for every $1\leq  r < k\leq M$ such that $\{\Lambda_r,\Lambda_k\}\in\mathcal{E}_M$, we have
\begin{equation}
(\iota_{\Lambda_r}^*\oplus\iota_{\Lambda_k}^*) Z (\iota_{\Lambda_r}\oplus\iota_{\Lambda_k}) \geq 0,
\end{equation}
meaning that
\begin{equation}
\begin{pmatrix}
\idty_{\Lambda_{r}} &(Z)_{r,k} \\
 (Z)_{r,k}^T         & \idty_{\Lambda_{k}}
\end{pmatrix} \geq 0.
\end{equation}
The latter is satisfied if and only if the corresponding  Schur complement is nonnegative, i.e.,
\begin{equation}
\idty_{\Lambda_r}- (Z)_{r,k} (Z)_{r,k}^T\geq 0.
\end{equation}
This proves that the norm of every off-diagonal block in $Z$ is bounded by $1$, i.e.,  
\begin{equation}
\big\|(Z)_{r,k}\big\|=\left\|h_{\Lambda_r}^{-1/2}(h_\Lambda)_{r,k}h_{\Lambda_k}^{-1/2}\right\| \leq 1, \quad\text{ for every } 1\leq  r , k\leq M;\ \{\Lambda_r,\Lambda_k\}\in\mathcal{E}_M.
\end{equation}
To summarize, the blocks of $Z$ satisfy
\begin{equation}
\left\|(Z)_{i,j}\right\|\ 
\begin{cases}
= 1 & \text{if } i=j\\
\leq 1 & \text{if }\{\Lambda_i,\Lambda_j\}\in\mathcal{E}_M\\
=0 & i\neq j \text{ and } \{\Lambda_i,\Lambda_j\}\not\in\mathcal{E}_M
\end{cases}.
\end{equation}
That the norm bound (\ref{eq:norm-bound-2}) is satisfied, follows then directly by observing that the block matrix inside the norm in (\ref{eq:norm-bound-2}) has nonzero entries only in at most $\Delta(\mathcal{G}_M)$ blocks in any (block-) row or column, where the norm of each block-nonzero-entry is less than or equal to one. Here we give a detailed proof. 

We write any given unit vector $x\in\R^{|\Lambda|}$ as
\begin{equation}
x=\sum_{m=1}^M e_m\otimes (x)_m
\end{equation}
where $(x)_m\in\R^{|\Lambda_m|}$ is the restriction of $x$ to  $\ell^2(\Lambda_m)$, i.e., $(x)_m:=\iota_{\Lambda_m}^* x$. Note that 
$
1=\|x\|^2=\sum_m\|(x)_m\|^2$.
It is easy to see that, using (\ref{eq:Z:1})
 \begin{equation}
 \langle x,(Z-\idty_\Lambda)x\rangle=
 \sum_{
 {\tiny
\begin{array}{c}
i,j \in\{1,\ldots,M\}\\
\{\Lambda_i,\Lambda_j\}\in\mathcal{E}_M
\end{array}}
 }
 \left\langle (x)_i,(Z)_{i,j}(x)_j\right\rangle.
 \end{equation}
 Thus,
 \begin{equation}\label{eq:norm-bound-3}
\left| \langle x,(Z-\idty_\Lambda)x\rangle\right|
 \leq
 \Big(\sum_{
 {\tiny
\begin{array}{c}
i,j\\
\{\Lambda_i,\Lambda_j\}\in\mathcal{E}_M
\end{array}}} \|(x)_i\|^2 \|(Z)_{i,j}\| \Big)^{\frac12}
 \Big(\sum_{
 {\tiny
\begin{array}{c}
i,j\\
\{\Lambda_i,\Lambda_j\}\in\mathcal{E}_M
\end{array}}}  \|(Z)_{i,j}\|\|(x)_j\|^2 \Big)^{\frac12}.
 \end{equation}
We then observe that
\begin{eqnarray}\label{eq:norm-bound-4}
\sum_{
 {\tiny
\begin{array}{c}
i,j\\
\{\Lambda_i,\Lambda_j\}\in\mathcal{E}_M
\end{array}}} \|(x)_i\|^2 \|(Z)_{i,j}\|
&\leq&
  \max_{i}\sum_{
j,\ 
\{\Lambda_i,\Lambda_j\}\in\mathcal{E}_M}\|(Z)_{i,j}\| \nonumber\\
&\leq&
\max_i\left|\{j;\ \{\Lambda_i,\Lambda_j\}\in\mathcal{E}_M\}\right| \nonumber\\
&=& \Delta(\mathcal{G}_M).
\end{eqnarray}
Similarly, we obtain
\begin{equation}\label{eq:norm-bound-5}
\sum_{
 {\tiny
\begin{array}{c}
i,j\\
\{\Lambda_i,\Lambda_j\}\in\mathcal{E}_M
\end{array}}}  \|(Z)_{i,j}\|\|(x)_j\|^2 
\leq
\max_j\left|\{i;\ \{\Lambda_i,\Lambda_j\}\in\mathcal{E}_M\}\right| 
= \Delta(\mathcal{G}_M).
\end{equation}
Substitute the bounds (\ref{eq:norm-bound-4}) and (\ref{eq:norm-bound-5}) in (\ref{eq:norm-bound-3}) to get the norm bound (\ref{eq:norm-bound-2}), which finishes the proof.

\end{proof}

We are now ready to finish the proof of the main theorem (\ref{thm:main-result}) starting with the entanglement bound (\ref{eq:LN:upperbound-0}).
A direct calculation using (\ref{def:M1}), (\ref{def:M2}), and (\ref{def:Tt-inv}) results the four blocks of  $[M^{(2)}_{t,\beta},\widetilde{\mathbb{P}}]\Pi_{J}(M^{(1)}_{t,\beta})$, that we will discuss one by one below. \\
\noindent\underline{The $(\cdot)_{1,1}$-block term:}
 \begin{eqnarray}\label{eq:4-blocks}
 \left(\cdot\right)_{1,1}&=&
 \left[\cos(2th_{\Lambda}^{1/2}),\mathbb{P}\right]\cos(2th_{\Lambda}^{1/2})+ \nonumber\\
&&\hspace{5cm}
+\left[h_{\Lambda}^{-1/2}\sin(2th_{\Lambda}^{1/2}),\mathbb{P}\right] h_{\Lambda}^{1/2} \sin(2th_{\Lambda}^{1/2}). 
 \end{eqnarray}
 We apply the inequality for the $\alpha$-quasi norms, see e.g., \cite[Thm 7.8]{Weidmann}
\begin{equation}\label{eq:alpha-norm-bound-0}
\|A_1+A_2\|_\alpha^\alpha\leq 2\|A_1\|_\alpha^\alpha+ 2\|A_2\|_\alpha^\alpha.
\end{equation}
Then we use the inequality
\begin{equation}\label{eq:alpha-norm-bound-1}
 \|A_1 A_2\|_\alpha\leq \|A_1\|_\alpha\ \|A_2\|
 \end{equation}
 that follows from the well known inequality $\sigma_j(A_1A_2)\leq \sigma_j(A_1)\|A_2\|$, see e.g., \cite{Bhatia},
to find that
\begin{equation}\label{eq:11-term-alpha-norm}
\left\|\left([M^{(2)}_{t,\beta},\widetilde{\mathbb{P}}]\Pi_{J}(M^{(1)}_{t,\beta})\right)_{1,1}\right\|_\alpha^\alpha\leq 2 \left\|\left[\cos(2th_{\Lambda}^{1/2}),\mathbb{P}\right]\right\|_\alpha^\alpha+2 \mathcal{C}_h^\alpha \left\|\left[h_{\Lambda}^{-1/2}\sin(2th_{\Lambda}^{1/2}),\mathbb{P}\right] \right\|_\alpha^\alpha.
\end{equation}
 $\|A\|_\alpha^\alpha$ can be bounded by the $\alpha$-power of the absolute values of its elements in any basis, see, e.g, the  proof of Lemma 2.1 in \cite{BSW19}, 
\begin{equation}\label{eq:alpha-norm-bound-2}
\|A\|_\alpha^\alpha\leq \sum_{j,k}|A_{ij}|^\alpha, \text{ for any }\alpha\in(0,1].
\end{equation}
We further note that for any operator $A$ acting on $\mathcal{H}_\Lambda$
\begin{equation}\label{eq:off-blocks}
|\langle\delta_x, [A,\mathbb{P}]\delta_y\rangle|=
\begin{cases}
0 & \text{if } x,y\in \Lambda_0 \text{ or } x,y\in\Lambda\setminus\Lambda_0\\
2|\langle\delta_x, A\delta_y\rangle| & \text{otherwise}
\end{cases}.
\end{equation}
By taking the supremum over $t$ and $\beta$ then averaging the disorder in (\ref{eq:11-term-alpha-norm}), (\ref{eq:off-blocks}) and (\ref{eq:alpha-norm-bound-2}) give
\begin{eqnarray}\label{eq:11-term-E}
\mathbb{E}\left(\sup_{t,\beta}\left\|\left(\cdot\right)_{1,1}\right\|_\alpha^\alpha\right)\nonumber
&\leq& 
8 \sum_{{\tiny
\left.\begin{array}{c}
 x\in\Lambda_0
 \\
 y\in\Lambda\setminus\Lambda_0
 \end{array}\right.
}}
\mathbb{E}\left(\sup_{|u|\leq 1}\left|\langle\delta_x, u(h_\Lambda)\delta_y\rangle\right|^{\alpha}\right)+\\
&&
\hspace{2cm}
 +8 \mathcal{C}_h^\alpha 
\sum_{{\tiny
\left.\begin{array}{c}
 x\in\Lambda_0
 \\
 y\in\Lambda\setminus\Lambda_0
 \end{array}\right.
}}
\mathbb{E}\left(\sup_{|u|\leq 1}\left|\langle\delta_x, h_{\Lambda}^{-1/2}u(h_\Lambda)\delta_y\rangle\right|^{\alpha}\right).
\end{eqnarray}
Here the operators $\sin(2th_\Lambda^{1/2})$ and $\cos(2t h_\Lambda^{1/2})$ are absorbed in $u(h_\Lambda)$.

We assume that
\begin{equation}\label{eq:alpha-s-1}
\alpha\in(0,s], \text{i.e., }0<\frac{\alpha}{s}\leq 1
\end{equation}
to use Jensen's inequality $\mathbb{E}(\cdot^{\alpha/s})\leq \mathbb{E}(\cdot)^{\alpha/s}$, and the eigenfunction correlators (\ref{def:EFC}) to bound the second sum in (\ref{eq:11-term-E})
\begin{eqnarray}\label{eq:11-term-2}
\sum_{{\tiny
\left.\begin{array}{c}
 x\in\Lambda_0
 \\
 y\in\Lambda\setminus\Lambda_0
 \end{array}\right.
}}
\mathbb{E}\left(\sup_{|u|\leq 1}\left|\langle\delta_x, h_{\Lambda}^{-1/2}u(h_\Lambda)\delta_y\rangle\right|^{\alpha}\right)&\leq& 
\sum_{{\tiny
\left.\begin{array}{c}
 x\in\Lambda_0
 \\
 y\in\Lambda\setminus\Lambda_0
 \end{array}\right.
}}
\mathbb{E}\left(\sup_{|u|\leq 1}\left|\langle\delta_x, h_{\Lambda}^{-1/2}u(h_\Lambda)\delta_y\rangle\right|^s\right)^{\alpha/s} \nonumber
\\
&\leq&  C^{\alpha/s} \sum_{{\tiny
\left.\begin{array}{c}
 x\in\Lambda_0
 \\
 y\in\Lambda\setminus\Lambda_0
 \end{array}\right.
}} e^{-\frac{\alpha}{s}\eta|x-y|}\nonumber
\\
&\leq& C^{\alpha/s}\left(\sum_{x\in\mathbb{Z}^d}e^{-\frac{\alpha}{s}\eta|x|}\right)^2|\partial\Lambda_0|.
\end{eqnarray}
In the second-to-third step we used the following argument: For each $x\in\Lambda_0$ and $y\in\Lambda\setminus\Lambda_0$ there exists at least one $z\in\partial\Lambda_0$ such that $|x-y|=|x-z|+|y-z|$, then
\begin{eqnarray}
\sum_{{\tiny
\left.\begin{array}{c}
 x\in\Lambda_0
 \\
 y\in\Lambda\setminus\Lambda_0
 \end{array}\right.
}} e^{-\frac{\alpha}{s}\eta|x-y|}
&\leq& \sum_{z\in\partial\Lambda_0} \sum_{
{\tiny
\begin{array}{c}
 x\in\Lambda_0,\ y\in\Lambda\setminus\Lambda_0
 \\
 |x-y|=|x-z|+|y-z|
 \end{array}
}} e^{-\frac{\alpha}{s}\eta|x-z|}\ e^{-\frac{\alpha}{s}\eta|y-z|} \nonumber\\
&\leq& \sum_{z\in\partial\Lambda_0} \left(\sum_{x\in\mathbb{Z}^d}e^{-\frac{\alpha}{s}\eta|x-z|}\right)^2.
\end{eqnarray}

Similarly, the first sum in the right hand side of (\ref{eq:11-term-E})  scales like $\mathcal{O}(|\partial\Lambda_0|)$, here we use the eigenfunction correlators (\ref{def:EFC-identity}),
\begin{eqnarray}\label{eq:11-term-1}
\sum_{{\tiny
\left.\begin{array}{c}
 x\in\Lambda_0
 \\
 y\in\Lambda\setminus\Lambda_0
 \end{array}\right.
}}
\mathbb{E}\left(\sup_{|u|\leq 1}\left|\langle\delta_x, u(h_\Lambda)\delta_y\rangle\right|^{\alpha}\right)&\leq& \sum_{{\tiny
\left.\begin{array}{c}
 x\in\Lambda_0
 \\
 y\in\Lambda\setminus\Lambda_0
 \end{array}\right.
}}
\mathbb{E}\left(\sup_{|u|\leq 1}\left|\langle\delta_x, u(h_\Lambda)\delta_y\rangle\right|^s\right)^{\alpha/s} \nonumber\\
&\leq& \mathcal{C}_h^{\alpha} C^{\alpha/s} \sum_{{\tiny
\left.\begin{array}{c}
 x\in\Lambda_0
 \\
 y\in\Lambda\setminus\Lambda_0
 \end{array}\right.
}} e^{-\frac{\alpha}{s}\eta|x-y|}\nonumber
\\
&\leq& \mathcal{C}_h^\alpha C^{\alpha/s}\left(\sum_{x\in\mathbb{Z}^d}e^{-\frac{\alpha}{s}\eta|x|}\right)^2|\partial\Lambda_0|.
\end{eqnarray}
Substitute (\ref{eq:11-term-1}) and (\ref{eq:11-term-2}) in (\ref{eq:11-term-E}) to find,
\begin{equation}\label{eq:11-term-F}
\mathbb{E}\left(\sup_{t,\beta}\left\|\left([M^{(2)}_{t,\beta},\widetilde{\mathbb{P}}]\Pi_{J}(M^{(1)}_{t,\beta})\right)_{1,1}\right\|_\alpha^\alpha\right)
\leq 16\mathcal{C}_h^\alpha C^{\alpha/s}\left(\sum_{x\in\mathbb{Z}^d}e^{-\frac{\alpha}{s}\eta|x|}\right)^2|\partial\Lambda_0|.
\end{equation}
\noindent\underline{The $(\cdot)_{1,2}$-block term:}
 \begin{eqnarray}
 \left(\cdot\right)_{1,2}&=& \left[\cos(2th_{\Lambda}^{1/2}),\mathbb{P} \right]h_{\Lambda}^{-1/2}\sin(2th_{\Lambda}^{1/2}) \mathcal{D}^+_\beta +\nonumber\\
 && \hspace{5cm}
 +(-1)\left[h_{\Lambda}^{-1/2}\sin(2th_{\Lambda}^{1/2}) ,\mathbb{P}\right] \cos(2th_{\Lambda}^{1/2}) \mathcal{D}^+_\beta.
  \end{eqnarray}
We will follow the same steps in bounding the $(\cdot)_{1,1}$ term. So we need to have only bounded terms outside of the commutators (with $\mathbb{P}$), so we push $h_\Lambda^{-1/2}$ in the first term of $(\cdot)_{1,2}$ inside the commutator to obtain the equivalent formula
\begin{eqnarray}\label{eq:12-term-0}
 (\cdot)_{1,2}&=&\left(\left[\cos(2th_{\Lambda}^{1/2})h_{\Lambda}^{-1/2},\mathbb{P} \right]-
 \cos(2th_{\Lambda}^{1/2})\left[h_{\Lambda}^{-1/2},\mathbb{P} \right]\right)\sin(2th_{\Lambda}^{1/2}) \mathcal{D}^+_\beta +\nonumber\\
 && \hspace{6.5cm} -\left[h_{\Lambda}^{-1/2}\sin(2th_{\Lambda}^{1/2}) ,\mathbb{P}\right] \cos(2th_{\Lambda}^{1/2}) \mathcal{D}^+_\beta.
 \end{eqnarray}
We proceed as in the $(\cdot)_{1,1}$ block-term case, i.e., we take the $\alpha$-quasi norm as in (\ref{eq:11-term-alpha-norm}) where we use $\|\mathcal{D}^+_\beta\|\leq \mathcal{C}_h$, then we use (\ref{eq:alpha-norm-bound-0}) iteratively for three terms.
We also
apply (\ref{eq:off-blocks}), then we finally take the supremum over $t\in\R$ and $\beta\in(0,\infty]^M$, and average to land in the bound
\begin{equation}\label{eq:12-term-F}
\mathbb{E}\left(\sup_{t,\beta}\left\|\left([M^{(2)}_{t,\beta},\widetilde{\mathbb{P}}]\Pi_{J}(M^{(1)}_{t,\beta})\right)_{1,2}\right\|_\alpha^\alpha\right)
\leq 40\mathcal{C}_h^\alpha C^{\alpha/s}\left(\sum_{x\in\mathbb{Z}^d}e^{-\frac{\alpha}{s}\eta|x|}\right)^2|\partial\Lambda_0|.
\end{equation}

\noindent\underline{The $(\cdot)_{2,1}$-block term:}
\begin{eqnarray}\label{eq:21-term-0}
  \left(\cdot\right)_{2,1}&=& \left[\mathcal{D}^-_\beta h_\Lambda^{1/2}\sin(2th_{\Lambda}^{1/2}) ,\mathbb{P}\right]\cos(2th_{\Lambda}^{1/2}) +\nonumber \\
  &&\hspace{5cm}
+(-1)\left[\mathcal{D}^-_\beta \cos(2th_{\Lambda}^{1/2}),\mathbb{P} \right]  h_\Lambda^{1/2} \sin(2th_{\Lambda}^{1/2}) .
  \end{eqnarray}
Here, we have the extra (not uniformly bounded) operator $\mathcal{D}^-_\beta$ inside both commutators. After taking the Schatten $\alpha$-quasi-norm to both sides of (\ref{eq:21-term-0}), we consider the supremum over $t$ and $\beta$, to obtain
\begin{eqnarray}
\sup_{t,\beta}\left\|\left([M^{(2)}_{t,\beta},\widetilde{\mathbb{P}}]\Pi_{J}(M^{(1)}_{t,\beta})\right)_{2,1}\right\|_\alpha^\alpha
&\leq& 
16\mathcal{C}_h^\alpha \sum_{{\tiny
\left.\begin{array}{c}
 x\in\Lambda_0
 \\
 y\in\Lambda\setminus\Lambda_0
 \end{array}\right.
}}
\sup_{\beta,\ |u|\leq 1}\left|\langle\delta_x, \mathcal{D}^-_\beta u(h_\Lambda)\delta_y\rangle\right|^{\alpha}.
\end{eqnarray}
Here, $h_\Lambda^{1/2}$ is absorbed in $|u|\leq 1$, i.e., 
\begin{equation}
\left|\langle\delta_x, \mathcal{D}^-_\beta h_\Lambda^{1/2} \sin(2th_{\Lambda}^{1/2})\delta_y\rangle\right|\leq \mathcal{C}_h \sup_{\beta,\ |u|\leq 1}\left|\langle\delta_x, \mathcal{D}^-_\beta u(h_\Lambda)\delta_y\rangle\right|.
\end{equation}

Then we proceed by expanding the multiplication of $\mathcal{D}^-_\beta$ and $u(h_\Lambda)$,
\begin{equation}\label{eq:21-term-22}
\leq 16 \mathcal{C}_h^\alpha \sum_{{\tiny
\left.\begin{array}{c}
 x\in\Lambda_0
 \\
 y\in\Lambda\setminus\Lambda_0
 \end{array}\right.
}}\ \sum_{z\in\Lambda} \sup_\beta\left|\langle\delta_x,\mathcal{D}^-_\beta\delta_z\rangle\rangle\right|^\alpha\ \sup_{|u|\leq 1}\left|\langle\delta_z,u(h_\Lambda)\delta_y\rangle\right|^\alpha
\end{equation}
where we used that $|\sum\cdot|^\alpha\leq \sum|\cdot|^\alpha$ for $\alpha\in(0,1]$.

Average the disorder, apply H\"{o}lder's inequality in (\ref{eq:21-term-22}), to obtain
\begin{eqnarray}
\mathbb{E}\left(\sup_{t,\beta}\left\|(\cdot)_{2,1}\right\|_\alpha^\alpha\right)
&\leq&
 16\mathcal{C}_h^\alpha\sum_{
{\tiny\begin{array}{c}
x\in\Lambda_0\\
y\in\Lambda\setminus\Lambda_0\\ 
z\in\Lambda
\end{array}}
} \mathbb{E}\left(\sup_\beta\left|\langle\delta_x,\mathcal{D}^-_\beta\delta_z\rangle\right|^{2\alpha}\right)^{\frac{1}{2}} \times \nonumber\\
&& \hspace{4cm}\times
\mathbb{E}\left( \sup_{|u|\leq 1}\left|\langle\delta_z,u(h_\Lambda)\delta_y\rangle\right|^{2\alpha}\right)^{\frac{1}{2}}
\end{eqnarray}
We assume that
\begin{equation}\label{eq:alpha-s-2}
\alpha\in (0,s/2]\ \text{ i.e., } \frac{2\alpha}{s}\leq 1
\end{equation}
to apply Jensen's inequality and get
\begin{equation}
\mathbb{E}\left(\sup_{t,\beta}\left\|(\cdot)_{2,1}\right\|_\alpha^\alpha\right)\leq 16\mathcal{C}_h^\alpha \sum_{
{\tiny\begin{array}{c}
x\in\Lambda_0\\
y\in\Lambda\setminus\Lambda_0\\ 
z\in\Lambda
\end{array}}
} \mathbb{E}\left(\sup_\beta\left|\langle\delta_x,\mathcal{D}^-_\beta\delta_z\rangle\right|^s\right)^{\alpha/s} 
\mathbb{E}\left(\sup_{|u|\leq 1} \left|\langle\delta_z,u(h_\Lambda)\delta_y\rangle\right|^s\right)^{\alpha/s}.
\end{equation}

Note that by assuming the validity of the eigenfunction correlators (\ref{def:EFC}) for all $a_1\leq b_1$, $a_2\leq b_2,\ldots$, and $a_d\leq b_d$ in $\Lambda=[a_1,b_1]\times \ldots\times[a_d,b_d]$, and by the translation invariance of the distribution of random parameters, (\ref{def:EFC}) also applies to the effective Hamiltonians $h_{\Lambda_m}^{-1/2}$ of the subsystems. Furthermore,  the hyperbolic tangent function falls within the set of functions $|u|\leq 1$. This leads to the decay bound
\begin{equation}\label{eq:EFC-D}
\mathbb{E}\left(\sup_\beta\left|\langle\delta_x,\mathcal{D}^-_\beta\delta_z\rangle\right|^s\right)\leq C e^{-\eta|x-z|}, \text{ for all }x,z\in \Lambda. 
\end{equation}
(\ref{def:EFC-identity}) and (\ref{eq:EFC-D}) give
\begin{eqnarray}\label{eq:21-term-F}
\mathbb{E}\left(\sup_{t,\beta}\left\|(\cdot)_{2,1}\right\|_\alpha^\alpha\right) &\leq&  16\mathcal{C}_h^{2\alpha} C^{2\alpha/s} \sum_{
x\in\Lambda_0,\ 
y\in\Lambda\setminus\Lambda_0,\
z\in\Lambda} e^{-\frac{\alpha}{s} \eta|x-z|}\ e^{-\frac{\alpha}{s} \eta|z-y|}\nonumber\\
&\leq& 16\mathcal{C}_h^{2\alpha} C^{2\alpha/s} \sum_{x\in\Lambda_0,\ y\in\Lambda\setminus\Lambda_0} e^{-\frac{\alpha}{2s} \eta|x-y|}\ \sum_{z\in\Lambda} e^{-\frac{\alpha}{2s} \eta|x-z|}\ e^{-\frac{\alpha}{2s}  \eta|z-y|}\nonumber\\
&\leq& 16\mathcal{C}_h^{2\alpha}C^{2\alpha/s}\left(\sum_{x\in\mathbb{Z}^d}e^{-\frac{\alpha}{2s} \eta|x|}\right)^2 
\left(\sum_{x\in\mathbb{Z}^d}e^{-\frac{\alpha}{s} \eta|x|}\right)
|\partial\Lambda_0|.
\end{eqnarray}
In the first-to-second step we used the basic inequality
\begin{equation}
|x-z|+|z-y|\geq \frac{1}{2}\left(|x-y|+|x-z|+|z-y|\right).
\end{equation}
\noindent\underline{The $(\cdot)_{2,2}$-block term:}

 \begin{eqnarray}
   \left(\cdot\right)_{2,2}&=& \left[\mathcal{D}^-_\beta h^{1/2}_{\Lambda}\sin(2th_{\Lambda}^{1/2}) ,\mathbb{P}\right]h_\Lambda^{-1/2}\sin(2th_{\Lambda}^{1/2}) \mathcal{D}^+_\beta + \nonumber\\
   && \hspace{5cm}
   +\left[\mathcal{D}^-_\beta \cos(2th_{\Lambda}^{1/2}),\mathbb{P} \right]\cos(2th_{\Lambda}^{1/2}) \mathcal{D}^+_\beta.
    \end{eqnarray}
Here we note that the first term has two problematic operators $\mathcal{D}^-_\beta$ and $h_\Lambda^{-1/2}$. This is making $(\cdot)_{2,2}$ the trickiest block in all the four blocks. 

Again here, we refrain from using the bound
\begin{equation}
h_\Lambda^{-1/2}\sin(2th_{\Lambda}^{1/2})\leq 2t \idty_\Lambda,
\end{equation}
that is producing a linear dependency on time.

As before, we rewrite the $(\cdot)_{2,2}$ term as
 \begin{eqnarray}
   \left(\cdot\right)_{2,2}&=& \left(\left[\mathcal{D}^-_\beta \sin(2th_{\Lambda}^{1/2}) ,\mathbb{P}\right]-\mathcal{D}^-_\beta h_\Lambda^{1/2}\sin(2th_{\Lambda}^{1/2})\left[ h_{\Lambda}^{-1/2} ,\mathbb{P}\right]\right)\sin(2th_{\Lambda}^{1/2}) \mathcal{D}^+_\beta + \nonumber\\
   && \hspace{5cm}
   +\left[\mathcal{D}^-_\beta \cos(2th_{\Lambda}^{1/2}),\mathbb{P} \right]\cos(2th_{\Lambda}^{1/2}) \mathcal{D}^+_\beta.
 \end{eqnarray}
The crucial fact here is that $\mathcal{D}^-_\beta h_\Lambda^{1/2}$ is uniformly bounded in the volume of the system, see Lemma \ref{lem:bound},
\begin{equation}
\|\mathcal{D}^-_\beta h_\Lambda^{1/2}\|\leq \sqrt{\Delta(\mathcal{G}_M)+1}.
\end{equation}

The same procedure as in the previous blocks gives that under the choice (\ref{eq:alpha-s-2}) on $\alpha$, we obtain the rough bound
\begin{equation}\label{eq:22-term-F}
\mathbb{E}\left(\sup_{t,\beta}\|(\cdot)_{2,2}\|_\alpha^\alpha\right)\leq 40\mathcal{C}_h^{2\alpha}C^{2\alpha/s}\left(\sum_{x\in\mathbb{Z}^d}e^{-\frac{\alpha}{2s} \eta|x|}\right)^3(\Delta(\mathcal{G}_M)+1)^{\alpha/2}\ |\partial\Lambda_0|.
\end{equation}

Finally, observe that our assumptions on the values of $\alpha$ in (\ref{eq:alpha-s-1}) and (\ref{eq:alpha-s-2}) allow us to choose $\alpha=s/2$.

Substitute the bounds for the four blocks (\ref{eq:11-term-F}), (\ref{eq:12-term-F}), (\ref{eq:21-term-F}), and (\ref{eq:22-term-F})  in (\ref{eq:LN-bound-E}) with $\alpha=s/2$ to obtain the  (rough) bound in Theorem \ref{thm:main-result}, with the pre-factor $C'$ defined in (\ref{def:C'}).



\begin{thebibliography}{99}
\bibitem{AbanicPepic17} D. Abanin and Z. Papic, \emph{Recent progress in many-body localization}, Annalen Der Physik \textbf{529} (2017), 1700169

\bibitem{AR18} H.~Abdul-Rahman, \emph{Entanglement of a class of non-gaussian states in disordered harmonic oscillator systems},  Journal of Mathematical Physics  \textbf{59} (2018), 031904

\bibitem{ARFS19} H. Abdul-Rahman, C. Fischbacher, and G. Stolz, \emph{Entanglement bounds in the XXZ quantum spin chain},  Annales Henri Poincar\'e \textbf{21} (2020), 2327-2366

\bibitem{ARNSS16} H. Abdul-Rahman, B. Nachtergaele, R. Sims, and G. Stolz, \emph{Entanglement dynamics of disordered quantum XY chains}, Letters in Mathematical Physics \textbf{106} (2016), 649–674

\bibitem{ARSS18} H. Abdul-Rahman, R. Sims, and G. Stolz, \emph{On the regime of localized excitations for the disordered oscillators}. Letters in Mathematical Physics  \textbf{110} (2020), 1159–1189



\bibitem{ARSS17} H. Abdul-Rahman, R. Sims, and G. Stolz, \emph{Correlations in disordered quantum harmonic oscillator systems: The effects of excitations and quantum quenches}, Contemporary Mathematics, \textbf{717} (2018), 31--47

\bibitem{ARS15} H. Abdul-Rahman and G. Stolz,  \emph{A uniform area law for the entanglement of eigenstates in the disordered XY chain},  Journal of Mathematical Physics  \textbf{56} (2015), 121901


\bibitem{Aizenman} M. Aizenman, J. H. Schenker, R. M. Friedrich, and D. Hundertmark,\emph{Finite-volume fractional-moment criteria for Anderson localization}, Communications in Mathematical Physics \textbf{224} (2001), 219–253

\bibitem{Aizenman-Warzel-15} M. Aizenman and S. Warzel, \emph{Random operators.  Disorder effects on quantum spectra and dynamics}, Graduate Studies in Mathematics, Vol.\textbf{168}, Am. Math. Soc., Providence (2015)

\bibitem{Agarwaletal} K. Agarwal, E. Altman, E. Demler, S. Gopalakrishnan, D. A. Huse, and M. Knap, \emph{Rare-region effects and dynamics near the many-body localization transition},  Annalen Der Physik (2017), 1600326

\bibitem{Alba-Calabrese17} V. Alba and P. Calabrese, \emph{Entanglement and thermodynamics after a quantum quench in integrable systems},
PNAS, \textbf{114} (30) (2017), 7947-7951

\bibitem{Alba-Calabrese18} V. Alba, P. Calabrese,
\emph{Entanglement dynamics after quantum quenches in generic integrable systems}, SciPost Physics \textbf{4} (2018), 017 

\bibitem{Altmam-etal-15} E. Altman and R. Vosk, \emph{Universal dynamics and renormalization in many body localized systems}, Annual Review of Condensed Matter Physics \textbf{6} (2015), 383--409

\bibitem{ALN2010} L. Amour, P. Levy-Bruhl, and J. Nourrigat, \emph{Dynamics and Lieb–Robinson estimates for lattices of interacting anharmonic oscillators.} 
Colloquium Mathematicum \textbf{118} (2010), 609–648

\bibitem{quasi-free3} H. Araki and M. Shiraishi, \emph{On quasifree states of the canonical commutation relations (I),} Publications of the Research Institute for Mathematical Sciences \textbf{7}, 105–120 (1971/72)

\bibitem{AudenaertEtal2002} K. Audenaert, J. Eisert, and M. B. Plenio, \emph{Entanglement properties of the harmonic chain}, Physical Review A  \textbf{66} (2002), 042327




\bibitem{BSW19} V.\ Beaud, J.\ Sieber, and S.\ Warzel, \emph{Bounds on the bipartite entanglement entropy for oscillator systems with or without disorder}, Journal of Physics A, \textbf{52} (2019), 235202


\bibitem{BW1} V. Beaud and S. Warzel, \emph{Low-energy Fock-space localization for attractive hard-core particles in disorder}, Annales Henri Poincar\'{e} \textbf{18} (2017) 3143--3166

\bibitem{BW2} V. Beaud and S. Warzel, \emph{Bounds on the entanglement entropy of droplet states in the XXZ spin chain}, Journal of Mathematical Physics \textbf{59} (2018), 012109

\bibitem{Bhatia} R. Bhatia, \emph{Matrix analysis}. Graduate Texts in Mathematics \textbf{169}, Springer, 1997

\bibitem{BH13} F. Brandao and M. Horodecki, \emph{An area law for entanglement from exponential decay of correlations}, Nature Physics \textbf{9} (2013), 721–726

\bibitem{BH15} F. Brandao and M. Horodecki, \emph{Exponential decay of correlations implies area law}, Communications in  Mathematical  Physics \textbf{333} (2015), 761–798

\bibitem{Bratteli-Robinson} O. Bratteli, D. Robinson, \emph{Operator algebras and quantum statistical mechanics 2}, 2nd ed. (Springer Verlag, NewYork, NY, 1997)

\bibitem{Symp1} L. Bruneau and J. Derezi\'{n}ski, \emph{Bogoliubov hamiltonians and one-parameter groups of Bogoliubov transformations}, Journal of Mathematical Physics \textbf{48} (2007), 022101


\bibitem{Calabrese06} P. Calabrese and J. Cardy, \emph{Time dependence of correlation functions following a quantum quench},
Physical Review Letters \textbf{96} (2006), 136801





\bibitem{Cr-Eis2006} M. Cramer, J. Eisert, \emph{Correlations, spectral gap, and entanglement in harmonic quantum systems on generic lattices}, New Journal of  Physics \textbf{8} (2016), 71

\bibitem{Cr-Se-Eis2008} M. Cramer, A. Serafini, J. Eisert, \emph{Locality of dynamics in general harmonic quantum systems}, Quantum information and many body quantum systems, pp. 51–73, CRM Series, Ed. Norm., Pisa (2008)

\bibitem{Williamson} de Gosson, M., \emph{Symplectic geometry and quantum mechanics}, Operator Theory: Advances and Applications (Birkh\"auser,Basel, 2006)


\bibitem{EisertEtal2010} J. Eisert, M. Cramer, and M. B. Plenio, \emph{Area laws for the entanglement entropy.} Reviews of Modern Physics \textbf{82} (2010), 277

\bibitem{EKS1} A. Elgart, A. Klein, and G. Stolz, \emph{Many-body localization in the droplet spectrum of the random XXZ quantum spin chain},  Journal of Functional Analysis \textbf{275} (2018),  211--258


\bibitem{EKS2} A. Elgart, A. Klein, and G. Stolz, \emph{Manifestations of dynamical localization in the disordered XXZ spin chain}, Communications in Mathematical Physics \textbf{361} (2018) 1083--1113
\bibitem{Fan51} K. Fan, \emph{Maximum properties and inequalities for the eigenvalues of completely continuous operators}, Proceedings of the National Academy of Sciences of the United States of America. \textbf{37} (1951), 760-766


\bibitem{Grimmett-etal-Ising} G. Grimmett, T. Osborne, and P. Scudo, \emph{Bounded entanglement entropy in the quantum Ising model}, Journal of Statistical Physics \textbf{178} (2020), 281–296



\bibitem{HSS12} E.~Hamza, R.~Sims, and G.~Stolz,
\emph{Dynamical localization in disordered quantum spin systems},
Communications in Mathematical Physics \textbf{315} (2012), 215--239.

\bibitem{Horn} A. Horn, \emph{On the singular values of a product of completely continuous operators}, Proceedings of the National Academy of Sciences of the United States of America \textbf{36} (1950), 374-375


\bibitem{Imbrie16} J. Z. Imbrie, \emph{On many-body localization for quantum spin chains}, Journal of Statistical Physics \textbf{163} (2016), 998--1048.

\bibitem{ImbrieEtal17} J. Z. Imbrie, V. Ros, and A. Scardicchio, \emph{Review:  Local  integrals  of  motion  in  many-body  localized systems}, Annalen der Physik \textbf{529} (2017), 1600278


\bibitem{QQ-Klobas21} K. Klobas and B. Bertini, \emph{Entanglement dynamics in Rule 54:  Exact results and quasiparticle picture}, arXiv:2104.04513v1

\bibitem{QQ-Lewis-Swan19} R. J. Lewis-Swan, A. Safavi-Naini, A. M. Kaufman and A. M. Rey, \emph{Dynamics of quantum information}, Nature Reviews Physics \textbf{1} (2019), 627–634

\bibitem{quasi-free2} J. Manuceau and A. Verbeure, \emph{Quasi-free states of the CCR algebra and Bogoliubov transformations,} Communications in Mathematical Physics \textbf{9} (1968), 293–302 

\bibitem{Mat-Ish1970} H. Matsuda and K. Ishii, \emph{Localization of normal modes and energy transport in the disordered harmonic chain.} Progress of Theoretical Physics Supplements  \textbf{45} (1970), 56–86

\bibitem{MS-Holstein-17} R. Mavi and J. Schencker, \emph{Localization in the disordered Holstein model}, Communications in Mathematical Physics \textbf{364} (2018), 719–764 

\bibitem{Messiah1999} Messiah, A., \emph{Quantum mechanics.} Dover, New York (1999)


\bibitem{QQ-Mitra18} A. Mitra, \emph{Quantum Quench Dynamics}, Annual Review of Condensed Matter Physics
\textbf{9} (2018), 245-259



\bibitem{MullerEtal20} P. M\"{u}ller, L. Pastur, and R. Schulte, \emph{How much delocalisation is needed for an enhanced area law of the entanglement entropy?} Communications in Mathematical Physics \textbf{376} (1) (2020), 649–679

\bibitem{NSS12} B. Nachtergaele, R. Sims, and G. Stolz, \emph{Quantum harmonic oscillator systems with disorder}, Journal of Statistical Physics \textbf{149} (2012), 969--1012


\bibitem{NSS13}  B.~Nachtergaele, R.~Sims and G.~Stolz, \emph{an area law for the bipartite entanglement of disordered harmonic oscillator systems}, Journal of Mathematical Physics \textbf{54} (2013), 042110

\bibitem{NH} R. Nandkishore and D. A. Huse, \emph{Many body localization and thermalization in quantum statistical mechanics}, Annual Review of Condensed Matter Physics \textbf{6} (2015), 15--38



\bibitem{PasturSlavin15} L. Pastur and V. Slavin, \emph{On the area law for disordered free fermions}, Physical Review Letters \textbf{113} (2014), 150404

\bibitem{Plenio2005} M. B. Plenio, \emph{The logarithmic negativity: A full entanglement monotone that is not convex}. Physical Review Letters \textbf{95} (2005), 090503 


\bibitem{PlenioEtal} M. B. Plenio, J. Hartley, J. Eisert, \emph{Dynamics and manipulation of entanglement in coupled harmonic systems with many degrees of freedom}, New Journal of Physics  \textbf{6} (2004), 36





\bibitem{Rangamani17} M. Rangamani, T. Takayanagi, \emph{Quantum Quenches and Entanglement}. In: Holographic Entanglement Entropy. Lecture Notes in Physics, vol 931. Springer, Cham. (2017)


\bibitem{ReedSimon} M. Reed, B. Simon, \emph{Methods of modern mathematical physics}, Academic Press, San Diego, 1975, Vol. 2


\bibitem{Sch-Cir-Wol2006} N. Schuch, J. I. Cirac, and M. Wolf, \emph{Quantum states on harmonic lattices}, Communications in  Mathematical Physics \textbf{267} (2006), 65–95 

\bibitem{SeiringerWarzel16} R. Seiringer and S. Warzel, \emph{Decay of correlations and absence of superfluidity in the disordered Tonks-Girardeau gas}, New Journal of Physics \textbf{18} (2016), 035002

\bibitem{Symp2} D. Shale, \emph{Linear symmetries of free boson fields}, Transactions of the American Mathematical Society \textbf{103} (1962), 149–167 

\bibitem{Simon} B. Simon, \emph{Trace ideals and their applications}. Mathematical Surveys and Monographs, Volume 120, 2005

\bibitem{SimsWarzel16} R. Sims and S. Warzel, \emph{Decay of determinantal and pfaffian correlation functionals in one-dimensional lattices}, Communications in  Mathematical Physics \textbf{347} (2016), 903--931

\bibitem{Stolz} G. Stolz, \emph{An introduction to the mathematics of Anderson localization}. Entropy and the quantum II, pp.71-108, Contemporary Mathematics \textbf{552}, American Mathematical Society, Providence, RI, 2011




\bibitem{VidalWerner} G. Vidal and R. Werner, \emph{Computable measure of entanglement}, Physical Review A \textbf{65}  (2002), 032314

\bibitem{Weidmann} J. Weidmann, \emph{Linear operators in Hilbert spaces}. Graduate Texts in Mathematics, vol. 68. Springer, Berlin (1980)

\bibitem{Weyl} H. Weyl, \emph{Inequalities between the two kinds of eigenvalues of a linear transformation}, Proceedings of the National Academy of Sciences of the United States of America  \textbf{35} (1949), 408–411

\end{thebibliography}
\end{document}